\documentclass[12pt]{article}

\usepackage[english]{babel}

\usepackage[a4paper,top=2cm,bottom=2cm,left=3cm,right=3cm,marginparwidth=1.75cm]{geometry}

\usepackage{amsmath}
\usepackage{graphicx}
\usepackage{yfonts}
\usepackage{amssymb}
\usepackage{amsthm,bm}
\usepackage{comment}
\usepackage{bbm}
\usepackage{mathtools}
\usepackage{enumerate}
\newtheorem{thm}{Theorem}[section]
\newtheorem{co}[thm]{Corollary}
\newtheorem{lem}[thm]{Lemma}

\newtheorem{assumption}[thm]{Assumption}

\newtheorem{pr}[thm]{Proposition}

\newtheorem{definition}[thm]{Definition}

\newtheorem{example}[thm]{Example}

\newtheorem{remark}[thm]{Remark}

\usepackage[colorlinks=true, allcolors=blue]{hyperref}

\newcommand{\E}{\mathbb{E}}

\newcommand{\R}{\mathbb R}

\newcommand{\IP}[2]{\langle #1,#2\rangle}

\title{Variational Characterization of Gaussian Feedback Capacity: Approximation and Attainment}
\author{\footnotesize \centering\begin{tabular}{ccc}
Jun Su & Yang Xu & Guangyue Han\\
The University of Hong Kong & Fudan University &The University of Hong Kong\\
 email:  junsu@hku.hk&email:  xuyyang@fudan.edu.cn  & email:  ghan@hku.hk\\
\end{tabular}}
\begin{document}
\maketitle

\begin{abstract}
We revisit the variational characterization of the feedback capacity of additive stationary Gaussian channels, formulated as a mixed optimization over an independent input spectrum and a strictly causal feedback filter. Under the general Paley--Wiener (Szeg\H{o}) condition, we prove that eliminating the independent input spectrum and restricting the feedback filter to real strictly causal polynomials of arbitrary finite degree preserves the feedback capacity at every average power constraint, thereby establishing the filter-only capacity formula in full generality. We next investigate attainment and uniqueness in the mixed variational problem: for every nonwhite noise spectrum satisfying the Paley--Wiener condition, the mixed supremum admits at most one maximizer; for rational noise spectra, attainment holds if and only if the spectrum is bounded away from zero, and in the nonwhite rational case every maximizer has zero independent input spectrum. Moving beyond capacity, we characterize exact output-spectrum preservation for any prescribed admissible pair: there exists a sequence of filter-only pairs reproducing its output spectrum exactly with input powers converging to that of the given pair if and only if the ratio of the independent input spectrum to the noise spectrum is integrable. When the noise spectrum is rational, the same integrability condition is necessary and sufficient for a single filter to preserve both the output spectrum and the input power exactly. Finally, we explicitly identify the unique capacity-achieving feedback filter for a class of stationary second-order autoregressive Gaussian channels.

\end{abstract}


\section{Introduction} \label{sec:introduction}
We consider the real additive Gaussian channel with noiseless feedback
\begin{equation}\label{real-ACGN-channel}
 Y_i=X_i(M,Y_1^{i-1})+Z_i,\qquad i=1,2,\ldots,
\end{equation}
where $M$ denotes the message to be communicated through the channel, the noise $\{Z_i\}$ is a zero-mean stationary Gaussian process independent of $M$, and the channel input $X_i$ at time $i$ is a function of $M$ and the past channel outputs $Y_1^{i-1}\triangleq (Y_1,\ldots,Y_{i-1})$. For a fixed $P>0$ and a code of blocklength~$n$, we impose the expected average-power constraint on the channel input $X=\{X_i\}$:
$$
\frac{1}{n}\sum_{i=1}^n\E[X_i^2]\le P.
$$
Let $C_\mathrm{FB}(P)$ denote the capacity of the channel \eqref{real-ACGN-channel}, commonly referred to as the {\em Gaussian feedback capacity} in the literature. Throughout the paper, all logarithms are natural, so rates are measured in nats per channel use.

Building on the finite-blocklength feedback capacity formula of Cover and Pombra \cite{cover1989gaussian}, Kim \cite[Theorem 3.2]{kim2010feedback} characterized $C_\mathrm{FB}(P)$ as the supremum of a stationary variational problem over pairs consisting of an independent input spectrum and a strictly causal feedback filter. Furthermore, \cite[Corollary 4.4]{kim2010feedback} asserts that at a maximizing pair the independent input component can be eliminated without changing the output spectrum or input power, leading to the filter-only capacity formula in \cite[Theorem 4.6]{kim2010feedback}. However, Derpich and \O stergaard \cite{derpich2022comments} identified a gap in the proof of \cite[Corollary 4.4]{kim2010feedback} concerning this component-removal step. Importantly, Kim's stationary mixed variational characterization \cite[Theorem 3.2]{kim2010feedback} is unaffected by this gap and serves as the starting point of the present work.

In this paper, we resolve this gap by establishing the filter-only capacity formula~\cite[Theorem~4.6]{kim2010feedback} directly under the general Paley--Wiener condition, without relying on exact component removal at a maximizer, and provide a complete picture of attainment, uniqueness, and exact spectrum preservation in Kim's mixed variational problem. Our main contributions are threefold:
\begin{enumerate}
  \item[1)] \emph{Attainment, uniqueness, and maximizer structure (Section~\ref{sec:attainment_and_uniqueness_in_kim_s_variational_problem}):} We prove that for any nonwhite noise spectrum satisfying the Paley--Wiener condition, the mixed variational problem admits at most one maximizer (Proposition~\ref{pr:nonwhite-optmimzer-existence-charaterziation}). For rational noise spectra, attainment holds if and only if the spectrum is bounded away from zero (Theorem~\ref{thm:optimizer-exist-rational}), and every maximizer in the nonwhite rational case has zero independent input spectrum (Proposition~\ref{pr:nonwhite-rational-optimizer-removal}).
  
  \item[2)] \emph{General filter-only capacity formula (Section~\ref{sec:main}):} Under the Paley--Wiener condition alone, we prove that the mixed and filter-only capacity suprema coincide for every $P > 0$, and that real strictly causal polynomials of arbitrary finite degree suffice to approach capacity (Theorem~\ref{thm:main}). The proof expresses the independent input spectrum in terms of a Schur function (Lemma~\ref{lamma-SY:label}), approximates this Schur function by real finite Blaschke products (Proposition~\ref{pr2-3:label}), and uses Beurling's theorem to construct strictly causal polynomial filters (Proposition~\ref{theorem:pw-counterpart}).

  \item[3)] \emph{Exact output-spectrum preservation (Section~\ref{sec:exact_output_spectrum_preservation}):} For any prescribed admissible pair, we prove that a sequence of filter-only pairs can reproduce its output spectrum exactly with asymptotic power preservation if and only if the ratio of the independent input spectrum to the noise spectrum is integrable (Theorem~\ref{app-cor:ess-inf-cod}), and show via conjugate-symmetric Blaschke interpolation that for rational spectra the same integrability condition is necessary and sufficient for a single filter to preserve both the output spectrum and the input power exactly (Theorem~\ref{thm:exact-removal-rational}).
\end{enumerate}
To illustrate these results, Section~\ref{sec:sk2-special_channels} explicitly identifies the unique optimal feedback filter for the feedback capacity of a class of stationary AR(2) (interleaved AR(1)) Gaussian channels (Proposition~\ref{prop:interleaved-ar1}).

We briefly situate our results within related work on rational and state-space Gaussian feedback channels. For nonwhite rational noise spectra bounded away from zero, Li and Elia~\cite{li2018youla} used the filter-only formula to compute bounds on feedback capacity and construct feedback codes. Sabag, Kostina, and Hassibi~\cite[Theorem~1]{sabag2023feedback} characterized the feedback capacity of multiple-input multiple-output (MIMO) Gaussian channels with state-space noise by a finite-dimensional convex optimization problem. For nonwhite rational noise under their state-space assumptions, Fay and Sabag \cite{fay2026feedback} proved that every optimizer of their finite-dimensional covariance problem has zero feedback-independent input variance and constructed a Schalkwijk--Kailath (SK) scheme that achieves every rate below feedback capacity. We return to the relation between their formulation and the mixed variational problem in Remark~\ref{remark:label-sec-2}.

The remainder of Section~\ref{sec:introduction} introduces the channel model and variational formulation (Section~\ref{sub:channel_model_and_variational_formulation}) along with the mathematical preliminaries and known optimality conditions (Section~\ref{subsection:preliminaries}). Following the main developments in Sections~\ref{sec:attainment_and_uniqueness_in_kim_s_variational_problem}--\ref{sec:sk2-special_channels} outlined above, Section~\ref{sec:conclusion} concludes the paper and discusses open problems, while the technical proofs of the polynomial approximation and finite Blaschke interpolation results are deferred to Appendices~\ref{appendix:Paley-Wiener Condition} and~\ref{appendix:lemma-Schur-completion}, respectively.

\subsection{Channel Model and Variational Formulation}\label{sub:channel_model_and_variational_formulation}
Let $\mathbb{D}=\{z\in\mathbb C:|z|<1\}$ and $\mathbb{T}=\{z\in\mathbb C:|z|=1\}$, with normalized measure $d\mu=d\theta/(2\pi)$. We assume that the noise process has a spectral density function (SDF) $S_Z\in L^1(\mathbb{T})$. Since the process is real, $S_Z$ is nonnegative and even: $S_Z(e^{i\theta})=S_Z(e^{-i\theta})$ almost everywhere. Throughout, we impose the Paley--Wiener condition, also called the Szeg\H{o} condition,
\begin{equation}\label{eq:PW-cond-All}
  \log S_Z\in L^1(\mathbb{T}).
\end{equation}
In particular, $S_Z>0$ almost everywhere on $\mathbb T$. We call $S_Z$ {\em nonwhite} if it is not constant almost everywhere on $\mathbb T$. Throughout the paper, we regard two spectral densities as equal if they agree almost everywhere on $\mathbb T$.

We consider real strictly causal filters of the form
\begin{equation}\label{eq:Bclass}
 B(z)=\sum_{k=1}^{\infty}b_kz^k\in z\mathcal{H}^2,\qquad b_k\in\R,
\end{equation}
where the Hardy-space notation ($\mathcal H^2$ and $z\mathcal H^2$) is recalled in Section~\ref{subsection:preliminaries}. We require finite weighted power 
$$
\int_{\mathbb{T}}|B|^2S_Z\,d\mu<\infty
$$
and let $X=V+BZ$, where $V=\{V_i\}$ is a zero-mean stationary Gaussian process independent of the noise process $\{Z_i\}$ and has a nonnegative even SDF $S_V\in L^1(\mathbb{T})$. Here $BZ$ denotes the mean-square spectral filtering of $Z$ with multiplier $B$. The channel input and output SDFs are
\begin{equation}\label{eq:sdfs-in-out-put}
 S_X=S_V+|B|^2S_Z,\qquad S_Y=S_V+|1+B|^2S_Z.
\end{equation}
Here $S_V$ represents the SDF of the feedback-independent input component, which may vanish identically. Define
\begin{align}
 \mathcal{P}(S_V,B)&=\int_{\mathbb{T}}(S_V+|B|^2S_Z)\,d\mu,\label{eq:power}\\
 \mathcal{R}(S_V,B)&=\frac12\int_\mathbb{T}\log\frac{S_Y}{S_Z}\,d\mu.\label{eq:rate}
\end{align}

We call $(S_V,B)$ {\em admissible} if $S_V\in L^1(\mathbb T)$ is nonnegative and even, $B\in z\mathcal{H}^2$ has real Taylor coefficients, and $\int_{\mathbb T}|B|^2S_Z\,d\mu<\infty$. An admissible pair is {\em feasible} at power $P$ if $\mathcal{P}(S_V,B)\le P$. A filter $B$ is called feasible at power $P$ if $(0, B)$ is feasible at power $P$.


Under \eqref{eq:PW-cond-All}, Kim's stationary variational characterization \cite[Theorem 3.2]{kim2010feedback} gives
\begin{equation}\label{eq:Kim}
 C_\mathrm{FB}(P)= \sup_{S_V,B}\mathcal{R}(S_V,B),
\end{equation}
where the supremum is taken over all feasible pairs $(S_V,B)$ at power $P$. Restricting \eqref{eq:Kim} to $S_V=0$ yields the filter-only supremum
\begin{equation}\label{eq:C-zero}
 C_0(P)=\sup_B\int_\mathbb{T}\log|1+B|\,d\mu,
\end{equation}
where the supremum is taken over all feasible filters $B$ at power $P$. Immediately from~\eqref{eq:Kim} and~\eqref{eq:C-zero}, we have
\begin{equation}\label{ine-C0-CFB}
C_0(P)\le C_\mathrm{FB}(P).
\end{equation}
The central question studied in this paper is whether equality holds in~\eqref{ine-C0-CFB} for every $P > 0$ under~\eqref{eq:PW-cond-All}, and when the suprema in~\eqref{eq:Kim} and~\eqref{eq:C-zero} are attained.

\subsection{Preliminaries}\label{subsection:preliminaries}
Let $f(z)=\sum_{k=0}^\infty c_kz^k$ be an analytic function on $\mathbb{D}$. We say that $f$ belongs to the class $\mathcal{H}^p=\mathcal{H}^p(\mathbb{D})$, $1\le p<\infty$, if 
$$
 \|f\|_{\mathcal{H}^p}=\sup_{0<r<1}\left(\int_{-\pi}^\pi|f(re^{i\theta})|^p\frac{d\theta}{2\pi}\right)^{1/p}<\infty.
$$
The space $\mathcal{H}^\infty$ consists of bounded analytic functions, with the supremum norm. Hardy functions have radial boundary values almost everywhere, denoted by the same symbol as the analytic function. For simplicity, we use $\| f\|_2=\| f\|_{\mathcal{H}^2}$, $\IP{f}{g}_{\mathcal{H}^2}=\int_\mathbb{T}f\overline g\,d\mu$, and $z\mathcal{H}^2=\{f\in \mathcal{H}^2:f(0)=0\}$. We say that $f$ has real Taylor coefficients if and only if $f(\bar z)=\overline{f(z)}$. In particular, for $f(z)\in\mathcal{H}^2$, it holds that
\begin{equation}
\| f\|_2^2 =\sum_{k\ge 0}|c_k|^2,\qquad |f(z)|\le \frac{\| f\|_2}{\sqrt{1-|z|^2}},
\end{equation}
where the first identity follows from \cite[Theorem 17.12]{rudin1987} and the second follows from the Cauchy--Schwarz inequality. In particular, convergence in $\mathcal{H}^2$ implies local uniform convergence in $\mathbb{D}$ \cite[Remark 17.8(c)]{rudin1987}.

For $f\in L^1(\mathbb T)$, we denote its $k$th Fourier coefficient by
$$
\widehat f(k)=\int_{\mathbb T} f(z)z^{-k}\,d\mu(z),\qquad k\in\mathbb Z.
$$
We call $f$ {\em causal} if $\widehat f(k)=0$ for every $k<0$. By the uniform density of trigonometric polynomials in the space of continuous functions on $\mathbb T$ \cite[Theorem 4.25]{rudin1987}, two functions in $L^1(\mathbb T)$ with the same Fourier coefficients agree almost everywhere.

An \emph{outer} function in $\mathcal{H}^p$ has the form
\begin{equation}\label{eq:outer-def}
 Q(z)=c\exp\left\{\int_{-\pi}^\pi\frac{e^{i\theta}+z}{e^{i\theta}-z} \log \omega(e^{i\theta})\frac{d\theta}{2\pi}\right\}, \quad |c|=1,
\end{equation}
where $\omega>0$ almost everywhere, $\omega\in L^p(\mathbb T)$, and $\log \omega\in L^1(\mathbb T)$. The outer function $Q(z)$ has no zeros in $\mathbb{D}$ and has boundary modulus $\omega$. The outer function $Q(z)$ is called normalized if $Q(0)>0$, equivalently $c=1$; see \cite[Sections 17.14--17.17]{rudin1987}.  In particular, if $S>0$ almost everywhere, $S\in L^1(\mathbb T)$, and $\log S\in L^1(\mathbb T)$, then its unique normalized outer spectral factor is
\begin{equation}\label{eq:outer-factor}
 \mathcal O_S(z)=\exp\left\{\frac12\int_{-\pi}^\pi \frac{e^{i\theta}+z}{e^{i\theta}-z}\log S(e^{i\theta})\frac{d\theta}{2\pi} \right\},
 \qquad |\mathcal O_S|^2=S \quad a.e. ~\mathrm{on}~\mathbb{T}.
\end{equation}
The function $\mathcal O_S(z)\in\mathcal{H}^2$ is often called the Szeg\H{o} function (see, e.g., \cite[Section 2.4]{simon2005}). If $S$ is even, then $\mathcal{O}_S$ has real Taylor coefficients. 

For the noise spectrum $S_Z$, it follows from \cite[Theorem 2.4.1]{simon2005} that, setting $H\triangleq\mathcal O_{S_Z}$, we have $S_Z=|H|^2$ and $H(0)>0$. We call the noise spectrum $S_Z$ \emph{rational} if its normalized outer factor $H = \mathcal{O}_{S_Z}$ is a rational function; by definition, every rational noise spectrum satisfies the Paley--Wiener condition~\eqref{eq:PW-cond-All}.

Fix an outer function $H\in\mathcal H^2$, and let $\mathbb C[z]$ denote the set of polynomials with complex coefficients. By Beurling's theorem \cite[Theorem~17.23]{rudin1987}, for every
$F\in\mathcal H^2$ and every $\varepsilon>0$, there exists $p\in\mathbb C[z]$ such that
\begin{equation}
    \|F-pH\|_2^2 = \int_{\mathbb T}|F-pH|^2\,d\mu < \varepsilon^2. \label{eq:outer-density}
\end{equation}
Equivalently, the polynomial multiples of $H$ are dense in $\mathcal H^2$.


A nonzero function $f\in\mathcal{H}^p$ has $\log|f|\in L^1(\mathbb{T})$ \cite[Theorem 17.17]{rudin1987}. The Smirnov class $\mathcal{N}^+$ has the quotient description
$$
 \mathcal{N}^+=\{a/b:a,b\in \mathcal{H}^\infty,\ b~~\text{is outer}\}.
$$
This is equivalent to its canonical-factorization definition in \cite[Section 2.5]{duren1970theory}. In particular, the canonical factorization theorem gives $\mathcal{H}^p\subset \mathcal{N}^+$, and the Smirnov theorem \cite[Theorem 2.11]{duren1970theory} states that a function in $\mathcal{N}^+$ with boundary values in $L^p$ belongs to $\mathcal{H}^p$. In addition, the maximum principle of Smirnov \cite[Lemma 2.3]{katsnelson1997theory} states that if $f\in \mathcal{N}^+$ and $|f|\le1$ almost everywhere on $\mathbb{T}$, then $f\in \mathcal{H}^\infty$ and $\|f\|_\infty\le1$.

For every admissible pair $(S_V,B)$, it follows from Hölder's inequality and $B,H\in\mathcal H^2$ that $BH\in\mathcal H^1\subset \mathcal N^+$. Moreover, $S_Z=|H|^2$ almost everywhere and admissibility gives
$$
\int_{\mathbb T}|BH|^2\,d\mu=\int_{\mathbb T}|B|^2S_Z\,d\mu<\infty.
$$
The Smirnov theorem therefore yields $BH\in\mathcal H^2$. Since $B(0)=0$, we conclude that $BH\in z\mathcal H^2$.

We say that $\phi$ is a Schur function if $\phi\in \mathcal{H}^\infty$ with $\|\phi\|_\infty\le1$; it is \emph{inner} if $|\phi|=1$ almost everywhere on $\mathbb{T}$. For $a\in\mathbb D$, let $b_a(z)=(z-a)/(1-\overline a z)$. Given $a_1,\ldots,a_N\in\mathbb D$ and $\zeta\in\mathbb T$, the function
$$
\Theta(z)=\zeta \prod_{j=1}^N b_{a_j}(z)
$$
is a \emph{finite Blaschke product} of degree $N$. The finite Blaschke product $\Theta $ is a rational inner function with zeros at the $a_j$, and nowhere else \cite[Theorem 15.21 and the subsequent discussion]{rudin1987}. The case $N=0$ is a unimodular constant (constant of modulus one).

We also use the following form of Jensen’s formula \cite[Theorem 15.18]{rudin1987}. 
\begin{lem}[{\cite{rudin1987}}]
Let $U\subset \mathbb C$ be an open subset containing $\overline{\mathbb{D}}$, and let $g:U\rightarrow \mathbb{C}$ be analytic with $g(0)\neq 0$. Denote by $z_1, \ldots , z_n$ the zeros of $g$ in $\mathbb{D}$, counted according to multiplicity. Then, we have
\begin{equation}\label{eq:jensen-formula}
\log |g(0)| = \sum_{k=1}^{n} \log (|z_k|) +  \int_{\mathbb T}\log |g| \, d\mu.
\end{equation}
\end{lem}
Note that any zeros of $g$ on $\mathbb{T}$ have integrable logarithmic singularities and contribute zero to \eqref{eq:jensen-formula} \cite[Lemma 15.17]{rudin1987}.

Finally, we recall from~\cite[Proposition~4.2(ii)--(iv)]{kim2010feedback} and~\cite[Appendix~A]{kim2010feedback} the necessary optimality conditions and a sufficient condition for attainment in the mixed variational problem~\eqref{eq:Kim}.
\begin{lem}[{\cite{kim2010feedback}}]
Suppose that the Paley--Wiener condition \eqref{eq:PW-cond-All} holds. Then the following statements hold:
\label{lem:kim-optimality}
\begin{enumerate}[(i)]
\item If a feasible pair $(S_V^\ast,B^\ast)$ attains the supremum in~\eqref{eq:Kim} for a given $P>0$, and $S_Y^\ast$ denotes its output spectrum, then
$$
    \eta \triangleq \operatorname*{ess\,inf}_{\mathbb T} S_Y^\ast > 0, \qquad S_V^\ast(S_Y^\ast-\eta)=0 \quad \text{a.e. on } \mathbb T,
$$
and $S_Y^\ast-(1+B^\ast)S_Z$ is causal. 
\item If $S_Z$ is bounded away from zero, i.e., $\operatorname*{ess\,inf}_{\mathbb T} S_Z>0$, then the supremum in~\eqref{eq:Kim} is attained for every $P>0$.
\end{enumerate}
\end{lem}


\section{Attainment, Uniqueness, and Maximizer Structure}\label{sec:attainment_and_uniqueness_in_kim_s_variational_problem}
We begin by analyzing the structure, uniqueness, and existence of maximizers in the mixed variational problem \eqref{eq:Kim}. Specifically, we show that a maximizing pair, whenever it exists, is unique for every nonwhite spectrum satisfying~\eqref{eq:PW-cond-All}, prove that the independent component vanishes at every maximizer for nonwhite rational noise, and derive a necessary and sufficient condition for attainment in the rational class.

\begin{pr}\label{pr:nonwhite-optmimzer-existence-charaterziation}
Suppose that the Paley--Wiener condition \eqref{eq:PW-cond-All} holds and $S_Z$ is nonwhite. Then, for every $P>0$, at most one pair feasible at power $P$ attains the supremum in \eqref{eq:Kim}. Moreover, if $S_Z$ is bounded away from zero, then for every $P > 0$, the supremum in \eqref{eq:Kim} is attained by a unique pair feasible at power $P$.
\end{pr}

\begin{proof}
Fix $P>0$, and let $(S_{V_i},B_i)$, $i=0,1$, be two pairs feasible at power $P$. For $i=0,1$, let $S_{Y_i}=S_{V_i}+|1+B_i|^2S_Z$ denote the corresponding output spectrum. For $0 < t < 1$, define a new pair $(S_{V_t},B_t)$ by
\begin{equation}\label{pr21-Bt-Svt-def}
B_t=(1-t)B_0+t B_1,\qquad S_{V_t}=(1-t)S_{V_0}+tS_{V_1}+D_t,
\end{equation}
where $D_t=t(1-t)|B_1-B_0|^2S_Z$. Note that for any $0<t<1$ and any $u_0,u_1\in \mathbb C$, it holds that
\begin{equation}\label{pr21-identity-u0-u1}
  (1 - t)|u_0|^2 + t|u_1|^2 - |(1 - t)u_0 + tu_1|^2 = t(1 - t)|u_1 - u_0|^2.
\end{equation}
Taking $u_i=1+B_i$ in \eqref{pr21-identity-u0-u1}, we can obtain 
\begin{equation}\label{pr21-output-specturm-S-Yt}
S_{Y_t}=(1-t)S_{Y_0}+tS_{Y_1}.
\end{equation}
Similarly, taking $u_i=B_i$ in \eqref{pr21-identity-u0-u1}, we have
\begin{equation}\label{pr21-Power-Svt-Bt}
\mathcal{P}(S_{V_t},B_t)=(1-t)\mathcal{P}(S_{V_0},B_0)+t\mathcal{P}(S_{V_1},B_1)\le P.
\end{equation}
By \eqref{pr21-Bt-Svt-def}, $S_{V_t}$ is nonnegative, even, and integrable, and $B_t\in z\mathcal H^2$ has real Taylor coefficients. Thus, \eqref{pr21-Power-Svt-Bt} shows that $(S_{V_t},B_t)$ is feasible at power $P$.

Suppose that $(S_{V_0},B_0)$ and $(S_{V_1},B_1)$ both attain the supremum in \eqref{eq:Kim}. We claim that
\begin{equation}\label{pr21-equivalence-S0-S1}
S_{Y_0}=S_{Y_1}\qquad\text{a.e. on }\mathbb T.
\end{equation}
To see this, let $E=\{S_{Y_0}\neq S_{Y_1}\}$ and suppose that $\mu (E)>0$. It then follows from the strict concavity of the logarithm that
\begin{equation}\label{pr21-identity-log-concave}
\log S_{Y_t}=\log \left[(1-t)S_{Y_0}+tS_{Y_1}\right]>(1-t)\log S_{Y_0} +t \log S_{Y_1}\quad \mathrm{on}~E,
\end{equation}
which further implies that
\begin{equation}
\begin{aligned}
\mathcal{R}(S_{V_t},B_t)&=\frac12 \int_{\mathbb T}\log\frac{ S_{Y_t}}{S_Z}d\mu\\
&=\frac12 \int_{\mathbb T\setminus E}\log\frac{ S_{Y_t}}{S_Z}d\mu+\frac12 \int_{E}\log\frac{ S_{Y_t}}{S_Z}d\mu\\
&>\frac12 \int_{\mathbb T\setminus E}\log\frac{ S_{Y_0}}{S_Z}\,d\mu+\frac12 \int_{E}(1-t)\log\frac{S_{Y_0}}{S_Z}\,d\mu+\frac12 \int_{E}t\log \frac{S_{Y_1}}{S_Z}\,d\mu\\
&=\frac12 (1-t)\int_{\mathbb T}\log\frac{ S_{Y_0}}{S_Z}\,d\mu +\frac12 t \int_{\mathbb T}\log\frac{ S_{Y_1}}{S_Z}\,d\mu\\
&=(1-t)\mathcal{R}(S_{V_0},B_0) +t\mathcal{R}(S_{V_1},B_1)\\
&=C_{\mathrm{FB}}(P).
\end{aligned}
\end{equation}
Since $(S_{V_t},B_t)$ is feasible at power $P$, this contradicts the definition of $C_{\mathrm{FB}}(P)$ in~\eqref{eq:Kim}. Hence, \eqref{pr21-equivalence-S0-S1} holds.

Note that if $B_0=B_1$ in $\mathcal{H}^2$, it immediately follows from \eqref{pr21-equivalence-S0-S1} that $S_{V_0}=S_{V_1}$ almost everywhere on $\mathbb T$. Suppose instead that $B_0\neq B_1$ in $\mathcal H^2$. First note that taking $t=1/2$ in \eqref{pr21-Bt-Svt-def} and by \eqref{pr21-output-specturm-S-Yt}--\eqref{pr21-equivalence-S0-S1}, we have the feasible pair $(S_{V_{1/2}}, B_{1/2})$ such that
$$
S_{Y_{1/2}} =\frac12 S_{Y_0} +\frac12 S_{Y_1}=S_{Y_0}=S_{Y_1},\qquad \mathcal{P}(S_{V_{1/2}}, B_{1/2})\le P,
$$
which implies that $(S_{V_{1/2}}, B_{1/2})$ is also optimal. By \cite[Theorem 17.17]{rudin1987}, it holds that $\log |B_1-B_0|\in L^1(\mathbb T)$, and thus $|B_1-B_0|>0$ almost everywhere on $\mathbb T$, which, together with \eqref{eq:PW-cond-All} and \eqref{pr21-Bt-Svt-def}, immediately implies that
\begin{equation}\label{pr21-Sv-mid-positive}
S_{V_{1/2}}\ge \frac14|B_1-B_0|^2S_Z>0\qquad a.e. ~\mbox{on }\mathbb T.
\end{equation}
Applying the water-filling condition in Lemma~\ref{lem:kim-optimality}(i) to the maximizing pair $(S_{V_{1/2}},B_{1/2})$, we obtain
\begin{equation}\label{pr21-water-filling}
S_{V_{1/2}}(S_{Y_{1/2}}-\eta)=0, \qquad a.e. \mbox{~on~} \mathbb T,
\end{equation}
where $\eta=\operatorname{ess\,inf}_{{\mathbb{T}}} S_{Y_{1/2}}$. Consequently, it follows from \eqref{pr21-Bt-Svt-def}, \eqref{pr21-Sv-mid-positive} and \eqref{pr21-water-filling} that
\begin{equation}\label{pr21-SY-mid-flat}
S_{Y_{1/2}}=\eta >0 \qquad a.e. \mbox{~on~} \mathbb T.
\end{equation}

Let $(S_V^\ast,B^\ast)$ be any pair attaining the supremum in \eqref{eq:Kim} with $S^\ast_Y=\eta>0$ almost everywhere on $\mathbb T$ and $m\triangleq\int_{\mathbb T}S_V^\ast\,d\mu>0$. We show that $S_Z$ is constant almost everywhere. To achieve this, let $(S_V,B)$ be any pair feasible at power $P$, with output spectrum $S_Y$. Applying \eqref{pr21-Bt-Svt-def}--\eqref{pr21-Power-Svt-Bt} to the two pairs by taking $(S_{V_0},B_0)=(S_{V}^\ast, B^\ast)$ and $(S_{V_1},B_1)=(S_V,B)$, we obtain a feasible pair $(S_{V_t},B_t)$ with output spectrum $(1-t)\eta+tS_Y$. Since $(S_V^\ast,B^\ast)$ is a maximizing pair in \eqref{eq:Kim}, it holds that
\begin{equation*}
\begin{aligned}
g(t)&\triangleq 2\left[\mathcal{R}(S_{V_t},B_t)-\mathcal{R}(S_{V}^\ast, B^\ast)\right]\\
 &=\int_{\mathbb{T}} \log\left(1 + t \left(\frac{S_Y}{\eta} - 1\right)\right) \,d\mu\le0.
\end{aligned}
\end{equation*}
Clearly, $g(0)=0$ and thus $g^\prime(0^+)\le 0$. Let $x=S_Y/\eta-1$. Then, $x\in L^1(\mathbb T)$ and $x\ge-1$. Moreover, since $\lvert\log(1+tx)\rvert/t\le 2|x|$ for $0<t\le1/2$, by the dominated convergence theorem, we have
$$
g^\prime(0^+)=\lim_{t\to 0^+}\frac{g(t)}{t}=\lim_{t\to 0^+}\frac{1}{t}\int_{\mathbb{T}} \log\left(1 + t \left(\frac{S_Y}{\eta} - 1\right)\right) \,d\mu=\int_{\mathbb{T}}\left(\frac{S_Y}{\eta} - 1 \right)\,d\mu.
$$
Consequently, it holds that
\begin{equation}\label{pr21-ine-Sy-S-mid-int}
\int_{\mathbb{T}}S_Y\,d\mu \le \eta=\int_{\mathbb{T}}S_{Y}^\ast\,d\mu.
\end{equation}

For each $k \ge 1$ and all sufficiently small real~$s$, define the pair $(\widetilde{S}_V,\widetilde{B})$ by
$$
\widetilde{B} = B^\ast + sz^k, \qquad \widetilde{S}_V= \left( 1 - \frac{\Delta_s}{m} \right) S_V^\ast,
$$
where 
$$
\Delta_s\triangleq\|\widetilde{B}H\|_2^2 - \|B^\ast H\|_2^2=2s \operatorname{Re}\langle B^\ast H, z^k H \rangle + s^2 \|H\|_2^2.
$$
By Section~\ref{subsection:preliminaries}, $B^\ast H\in z\mathcal H^2$. Hence, $\widetilde B H=B^\ast H+sz^kH\in z\mathcal H^2$. Also, we have
\begin{equation}\label{pr21-power-equiv-tilta-mid}
\begin{aligned}
  \mathcal{P}(\widetilde{S}_V,\widetilde{B})&=\int_{\mathbb T}\left(\widetilde{S}_V+|\widetilde{B}|^2S_Z\right)\,d\mu\\
  &=\int_{\mathbb T} \left( 1 - \frac{\Delta_s}{m} \right) S_{V}^\ast\,d\mu+\| \widetilde{B}H\|^2_2\\
  &=\int_{\mathbb T}S_V^\ast\,d\mu - \frac{\Delta_s}{m}\int_{\mathbb T}S_{V}^\ast\,d\mu + \| \widetilde{B}H\|^2_2\\
  &= m - \Delta_s + \| \widetilde{B}H\|^2_2 \\
  &=m + \| B^\ast H\|_2^2\\
  &=\mathcal{P}(S_V^\ast,B^\ast).
\end{aligned}
\end{equation}
Since $m>0$ and $\lim_{s\to 0}\Delta_s=0$, we have $\widetilde S_V\ge0$ for all sufficiently small $|s|$. Thus, by~\eqref{pr21-power-equiv-tilta-mid}, $(\widetilde S_V,\widetilde B)$ is feasible at power $P$ for both signs of $s$. Moreover, since $|1 + \widetilde{B}|^2 = |\widetilde{B}|^2 + 1 + 2 \operatorname{Re} \widetilde{B}$, it holds that
\begin{equation}\label{pr21-int-Sy-tilta}
\begin{aligned}
\int_{\mathbb T} \widetilde{S}_Y\,d\mu &= \int_{\mathbb T} \widetilde{S}_V+|1+\widetilde{B}|^2S_Z\,d\mu\\
&=\int_{\mathbb T} \left(\widetilde{S}_V+|\widetilde{B}|^2S_Z +(1 + 2 \operatorname{Re} \widetilde{B}) S_Z\right) \, d\mu\\
&=\mathcal{P}(\widetilde{S}_V,\widetilde{B})+\int_{\mathbb{T}} (1 + 2 \operatorname{Re} \widetilde{B}) S_Z \, d\mu,
\end{aligned}
\end{equation}
where $\widetilde{B}S_Z=(\widetilde{B}H)\overline{H}\in L^1(\mathbb T)$. Similarly, we have
\begin{equation}\label{pr21-int-Sy-mid}
\int_{\mathbb T} S_Y^\ast\,d\mu = \mathcal{P}(S_{V}^\ast,B^\ast)+\int_{\mathbb{T}} (1 + 2 \operatorname{Re} B^\ast) S_Z \, d\mu.
\end{equation}
Thus, combining \eqref{pr21-power-equiv-tilta-mid}--\eqref{pr21-int-Sy-mid}, we have
\begin{equation}\label{pr21-diff-Sy-tilta-mid}
\int_{\mathbb T} ( \widetilde{S}_Y- S_{Y}^\ast) \, d\mu = 2s \operatorname{Re} \int_{\mathbb T} z^k S_Z \, d\mu.
\end{equation}
By \eqref{pr21-ine-Sy-S-mid-int}, it holds that
$$
\int_{\mathbb T}\widetilde S_Y\,d\mu\le\eta=\int_{\mathbb T}S_Y^\ast\,d\mu.
$$
Applying \eqref{pr21-diff-Sy-tilta-mid} with both signs of $s$, we obtain that $\operatorname{Re} \int_{\mathbb T} z^k S_Z \, d\mu=0$. Since $S_Z$ is real and even, it holds that
\begin{equation}\label{pr21-int-zk-sz-0}
\int_{\mathbb T} z^k S_Z \, d\mu=0
\end{equation}
for all $k\ge 1$. By~\eqref{pr21-int-zk-sz-0}, $\widehat{S_Z}(-k)=0$ for every $k\ge1$. Since $S_Z$ is real-valued, $\widehat{S_Z}(k)=\overline{\widehat{S_Z}(-k)}=0$ for every $k\ge1$. Thus, by the uniqueness of Fourier coefficients of $S_Z$ in $L^1(\mathbb T)$, we obtain
$$
S_Z(e^{i\theta}) =\widehat{S_Z}(0)=\int_{\mathbb T}S_Z\,d\mu
$$
almost everywhere on $\mathbb T$. 

Now, returning to the maximizing midpoint pair in \eqref{pr21-Bt-Svt-def}, we have $\int_{\mathbb T}S_{V_{1/2}}\,d\mu>0$ by~\eqref{pr21-Sv-mid-positive} and $S_{Y_{1/2}}=\eta$ almost everywhere by \eqref{pr21-SY-mid-flat}. Applying \eqref{pr21-ine-Sy-S-mid-int}--\eqref{pr21-int-zk-sz-0} to the maximizing midpoint pair $(S_{V_{1/2}},B_{1/2})$, we obtain $\widehat{S_Z}(k)=0$ for every $k\ne0$. By uniqueness of Fourier coefficients, $S_Z$ is constant almost everywhere on $\mathbb T$, contrary to the nonwhite assumption. Hence, $B_0=B_1$, and \eqref{pr21-equivalence-S0-S1} establishes $S_{V_0}=S_{V_1}$ almost everywhere. This proves uniqueness. 

Finally, if $S_Z$ is bounded away from zero, existence follows from Lemma~\ref{lem:kim-optimality}(ii).
\end{proof}

The proof of Proposition~\ref{pr:nonwhite-optmimzer-existence-charaterziation} also yields the following implication.
\begin{remark}\label{rem:flat-optimal-output}
Suppose that the Paley--Wiener condition~\eqref{eq:PW-cond-All} holds and fix $P>0$. If a pair $(S_V^\ast, B^\ast)$ feasible at power $P$ attains the supremum in~\eqref{eq:Kim} with output spectrum $S_Y^\ast$ that is constant almost everywhere on $\mathbb{T}$ and satisfies $\int_{\mathbb{T}} S_V^\ast\, d\mu > 0$, then $S_Z$ is constant almost everywhere on $\mathbb{T}$. Equivalently, when $S_Z$ is nonwhite, every maximizer of~\eqref{eq:Kim} with an almost everywhere constant output spectrum satisfies $S_V^\ast = 0$ almost everywhere on $\mathbb{T}$. This establishes the flat-output case of~\cite[Remark~4.5]{kim2010feedback}.
\end{remark}


We now specialize to rational noise spectra. Write the normalized outer factor as $H = N_H/D_H$, where $N_H$ and $D_H$ are coprime real polynomials and $D_H(0) = 1$. Then, $N_H(0) = H(0) > 0$. Since $H$ is outer, $N_H$ has no zeros in $\mathbb{D}$; since $H \in \mathcal{H}^2$, $D_H$ has no zeros in $\overline{\mathbb{D}}$. Consequently, $S_Z = |H|^2$ is bounded away from zero if and only if $N_H$ has no zeros on $\mathbb T$. We next show that for nonwhite rational noise, no maximizer of \eqref{eq:Kim} can have a nontrivial independent input component.

\begin{pr}\label{pr:nonwhite-rational-optimizer-removal}
Suppose that $S_Z$ is rational and nonwhite. Then for every $P > 0$, any pair $(S_V^\ast, B^\ast)$ feasible at power $P$ attaining the supremum in~\eqref{eq:Kim} satisfies $S_V^\ast = 0$ almost everywhere on $\mathbb{T}$.
\end{pr}

\begin{proof}
Suppose that $(S_V^\ast,B^\ast)$ attains the supremum in~\eqref{eq:Kim}, and let $S_Y^\ast$ denote its output spectrum. By the output-spectrum and water-filling conditions in Lemma~\ref{lem:kim-optimality}(i),
\begin{equation}\label{pr23-kimpr42-1}
\eta\triangleq\operatorname*{ess\,inf}_{\mathbb{T}} S_Y^\ast> 0, \qquad S_V^\ast(S_Y^\ast - \eta) = 0 \quad \text{a.e.}
\end{equation}
and 
$$
K \triangleq S^\ast_ Y - (1 + B^\ast)S_Z=S_V^\ast+\overline{B^\ast} (1+B^\ast)S_Z
$$ 
is causal. Suppose that $m=\int_{\mathbb T} S_V^\ast\,d\mu>0$, and let $E\triangleq\{S_V^\ast>0\}$. Then, $\mu(E)>0$. By~\eqref{pr23-kimpr42-1}, it holds that
\begin{equation}\label{pr23-Sy-flat-on-E}
S_Y^\ast=\eta\qquad \text{a.e. on }E.
\end{equation}

Let $d=\max\{\deg N_H,\deg D_H\}$. Since $S_Z=|N_H/D_H|^2$, we have
\begin{equation}\label{pr23-T-def}
T\triangleq |D_H|^2(S^\ast_Y-\eta)=|D_H|^2(K-\eta)+|N_H|^2(1+B^\ast).
\end{equation}
It is straightforward to see that the Fourier coefficients of $|N_H|^2$ and $|D_H|^2$ vanish outside $[-d,d]$. Since both $K-\eta$ and $1+B^\ast$ are causal, by~\eqref{pr23-T-def}, it holds that $\widehat T(j)=0$ for $j<-d$. Moreover, since $T$ is real-valued, $\widehat T(j)=\overline{\widehat T(-j)}$, which implies that $\widehat T(j)=0$ also for $j>d$. By uniqueness of Fourier coefficients in $L^1(\mathbb T)$, we have
$$
T(z)=\sum_{j=-d}^{d}\widehat T(j)z^j \qquad\text{almost everywhere on }\mathbb T.
$$
Thus, $z^dT(z)$ agrees almost everywhere on $\mathbb T$ with a polynomial of degree at most $2d$. However, it follows from~\eqref{pr23-Sy-flat-on-E} and \eqref{pr23-T-def} that $T=0$ on $E$ almost everywhere. Consequently, $T=0$ almost everywhere on $\mathbb T$. Since $D_H(z)$ has no zeros on $\mathbb T$ and by~\eqref{pr23-T-def}, we obtain
$$
S^\ast_Y=\eta\qquad\text{a.e. on~}\mathbb T.
$$
Since $m>0$, Remark~\ref{rem:flat-optimal-output} implies that $S_Z$ is constant almost everywhere, contrary to the nonwhite assumption. Therefore, $m=\int_{\mathbb T}S_V^\ast\,d\mu=0$, and the nonnegativity of $S_V^\ast$ implies that $S_V^\ast=0$ almost everywhere on $\mathbb T$.
\end{proof}

The following theorem shows that within the class of rational spectra, the sufficient condition for attainment in Lemma~\ref{lem:kim-optimality}(ii) is also necessary, thereby yielding a complete characterization of attainment in \eqref{eq:Kim}.

\begin{thm}\label{thm:optimizer-exist-rational}
Suppose that $S_Z$ is rational. Then for every $P > 0$, the supremum in~\eqref{eq:Kim} is attained if and only if $S_Z$ is bounded away from zero.
\end{thm}


\begin{proof}
If $S_Z$ is bounded away from zero, the existence of the optimal pair follows from Lemma~\ref{lem:kim-optimality}(ii).

Conversely, suppose that $(S_V^\ast,B^\ast)$ attains the supremum in~\eqref{eq:Kim}. If $S_Z$ is constant almost everywhere, it then follows from \eqref{eq:PW-cond-All} that this constant is positive, and the conclusion follows. We may therefore assume that $S_Z$ is nonwhite. By~Proposition~\ref{pr:nonwhite-rational-optimizer-removal}, $S_V^\ast=0$ almost everywhere on $\mathbb T$. Since $(S_V^\ast,B^\ast)$ attains the supremum in~\eqref{eq:Kim}, the output-spectrum condition in Lemma~\ref{lem:kim-optimality}(i) shows that 
$$
\eta^\ast\triangleq\operatorname*{ess\,inf}_{\mathbb T}S_Y^\ast>0.
$$
Hence,
\begin{equation}\label{pf-L1-thm:optimizer-exist-rational}
\frac{\eta^\ast}{S_Z} \le \frac{S_Y^\ast}{S_Z} =|1+B^\ast|^2\in L^1(\mathbb T).
\end{equation}
The rational function $H$ is analytic in a neighborhood of $\overline{\mathbb D}$. If $H$ had a zero of order $r\ge1$ at $z_0=e^{i\theta_0}$, then for $z$ near $z_0$,
$$
H(z)=(z-z_0)^rh(z),
$$
where $h(z)$ is analytic and $h(z_0)\neq 0$. Since $|e^{i\theta}-z_0|\le |\theta-\theta_0|$ and $h(z)$ is bounded near $z_0$, it holds that
$$
\frac{1}{S_Z(e^{i\theta})}\ge c|\theta-\theta_0|^{-2r}
$$
for some $c>0$ and all sufficiently small $|\theta-\theta_0|>0$. This contradicts \eqref{pf-L1-thm:optimizer-exist-rational}. Hence, $H$ has no zeros on $\mathbb T$, and continuity on $\mathbb T$ implies that $\min_{z\in\mathbb T}|H(z)|^2>0$. Hence, $S_Z$ is bounded away from zero.
\end{proof}

Combining Proposition~\ref{pr:nonwhite-optmimzer-existence-charaterziation}, Proposition~\ref{pr:nonwhite-rational-optimizer-removal}, and Theorem~\ref{thm:optimizer-exist-rational} immediately establishes Kim's component-removal assertion \cite[Corollary 4.4]{kim2010feedback} for nonwhite rational spectra bounded away from zero.

\begin{co}\label{co:label-sec-2}
Suppose that $S_Z$ is nonwhite, rational, and bounded away from zero. Then for every $P>0$, the supremum in \eqref{eq:Kim} is attained by a unique feasible pair $(0,B^\ast)$. Moreover, $B^\ast$ attains the supremum in \eqref{eq:C-zero}, and $C_\mathrm{FB}(P)=C_0(P)$.
\end{co}

\begin{proof}
Existence follows from Theorem~\ref{thm:optimizer-exist-rational} and uniqueness follows from Proposition~\ref{pr:nonwhite-optmimzer-existence-charaterziation}. By Proposition~\ref{pr:nonwhite-rational-optimizer-removal}, the maximizing pair has the form $(0,B^\ast)$. Thus, $B^\ast$ is feasible in \eqref{eq:C-zero}
 and achieves $C_\mathrm{FB}(P)$. Together with \eqref{ine-C0-CFB}, this proves that $C_0(P)=C_\mathrm{FB}(P)$ and that $B^\ast$ attains the supremum in \eqref{eq:C-zero}. Taking $B^\ast$ itself as the filter in \cite[Corollary 4.4]{kim2010feedback} preserves both the output spectrum and the input power.
\end{proof}

\begin{remark}\label{remark:label-sec-2}
The results of this section refine and complement several findings in the literature:
\begin{itemize}
\item \emph{Relation to Kim~\cite{kim2010feedback}:} For nonwhite rational $S_Z$, Proposition~\ref{pr:nonwhite-rational-optimizer-removal} establishes the assertion in \cite[Remark~4.5]{kim2010feedback} that every maximizer $(S_V^\ast,B^\ast)$ satisfies $S_V^\ast=0$ almost everywhere on $\mathbb T$. Within the rational class, Theorem~\ref{thm:optimizer-exist-rational} provides the converse to Kim's sufficient condition for attainment \cite[Appendix~A]{kim2010feedback}.

\item \emph{Relation to Liu and Han~\cite{liu2018feedback}:} For nonwhite $S_Z$ bounded away from zero, Liu and Han \cite[Theorem 11]{liu2018feedback} used the change of variable $C=BH$ to show that any two optimal filters in the filter-only problem \eqref{eq:C-zero} coincide. Proposition~\ref{pr:nonwhite-optmimzer-existence-charaterziation} strengthens this to uniqueness of the entire maximizing pair $(S_V^\ast,B^\ast)$ in the mixed problem \eqref{eq:Kim}, whenever attainment holds, for every nonwhite $S_Z$ satisfying~\eqref{eq:PW-cond-All}.

\item \emph{Relation to Fay and Sabag~\cite{fay2026feedback}:} For nonwhite rational noise under their state-space assumptions, Fay and Sabag \cite[Theorem~1]{fay2026feedback} proved that every optimizer of their finite-dimensional covariance problem satisfies $M=0$, where $M$ is the variance of the feedback-independent Gaussian input component. For an optimizer satisfying the Riccati constraint with equality, their argument following \cite[eq.~(15)]{fay2026feedback} places the stationary input in the mean-square closed linear span of the past noise samples, yielding \(S_V=0\) in that representation. Notably, their model allows $S_Z$ to vanish on $\mathbb T$; when a rational $S_Z$ has a zero on $\mathbb T$, Theorem~\ref{thm:optimizer-exist-rational} shows that no feasible pair $(S_V,B)$ with $B\in z\mathcal H^2$ attains the supremum in \eqref{eq:Kim}.
\end{itemize}
\end{remark}

\section{The Filter-Only Capacity Formula}\label{sec:main}
Corollary~\ref{co:label-sec-2} established the existence of a filter-only maximizer for nonwhite rational spectra bounded away from zero. We now show that the mixed and filter-only capacity suprema coincide for every noise spectrum satisfying the Paley--Wiener condition \eqref{eq:PW-cond-All}, even when neither supremum is attained.

\begin{thm}\label{thm:main}
Suppose that the Paley--Wiener condition \eqref{eq:PW-cond-All} holds. Then for every $P>0$,
\begin{equation}\label{eq:main}
    C_\mathrm{FB}(P) =C_0(P).
\end{equation}
Moreover, the supremum in \eqref{eq:C-zero} remains unchanged when $B$ is restricted to real strictly causal polynomials of arbitrary finite degree.
\end{thm}

The proof of Theorem~\ref{thm:main} proceeds in three steps. First, Lemma~\ref{lamma-SY:label} identifies the independent input spectrum with the defect of a Schur function $\phi$ and expresses both the rate and input power in terms of outer spectral factors. Second, Proposition~\ref{pr2-3:label} approximates $\phi$ by real finite Blaschke products while preserving the output spectrum and recovering the input power asymptotically. Third, Proposition~\ref{theorem:pw-counterpart} (proved in Appendix~\ref{appendix:Paley-Wiener Condition}) converts these inner approximants into strictly causal polynomial filters with arbitrarily small power excess and rate loss. Theorem~\ref{thm:main} then follows by applying a preliminary power backoff to any feasible mixed pair before invoking Proposition~\ref{theorem:pw-counterpart}.

\subsection{Spectral Factorization and the Schur Defect}
\begin{lem}\label{lamma-SY:label}
Suppose that the Paley--Wiener condition \eqref{eq:PW-cond-All} holds, and let $(S_V,B)$ be an admissible pair with output spectrum $S_Y$. Set $H=\mathcal{O}_{S_Z}$ and $F=(1+B)H$. Then $F\in \mathcal{H}^2$, $S_Y$ admits a unique normalized outer spectral factor $G=\mathcal{O}_{S_Y}\in \mathcal{H}^2$, and $\phi=F/G$ is a Schur function. Furthermore, $H,G,F$, and $\phi$ have real Taylor coefficients, and
\begin{align}
 S_V&=|G|^2(1-|\phi|^2)\quad\mathrm{a.e.~on~}\mathbb T,\label{eq:schur-defect}\\
 \mathcal{R}(S_V,B)&=\log \frac{G(0)}{H(0)}=-\log\phi(0)\in[0,\infty),\label{eq:rate-constant}\\
 \mathcal{P}(S_V,B)&=\| G\|_2^2+\| H\|_2^2 -2\operatorname{Re}\IP{\phi G}{H}_{\mathcal{H}^2}.\label{eq:power-affine}
\end{align}
\end{lem}

\begin{proof}
By Section~\ref{subsection:preliminaries}, we have $BH\in z\mathcal H^2$. Hence, $F=H+BH\in\mathcal H^2$ and $F(0)=H(0)>0$. Then, $F$ is not identically zero. Since \(S_V\in L^1(\mathbb T)\), we have
\begin{equation}\label{SY-L1}
S_Y=S_V+|1+B|^2S_Z=S_V+|F|^2\in L^1(\mathbb T).
\end{equation}
Then, by \cite[Theorem 17.17]{rudin1987}, it holds that
\begin{equation}\label{log-F-L1}
\log |F(e^{i\theta})|\in L^1,
\end{equation}
which shows that $|F(e^{i\theta})|>0$ almost everywhere. Since 
\begin{equation}\label{SY-F}
S_Y(e^{i\theta})=S_V(e^{i\theta}) + |F(e^{i\theta})|^2 \ge|F(e^{i\theta})|^2,
\end{equation}
it follows that $S_Y(e^{i\theta}) >0$ almost everywhere. It then follows from \eqref{log-F-L1} and \eqref{SY-F} that
$$
\log^-S_Y(e^{i\theta}) \le 2\log^-|F(e^{i\theta})|\in L^1,
$$
which, together with the fact that $\log^+S_Y(e^{i\theta})\le S_Y(e^{i\theta})$ and \eqref{SY-L1}, immediately implies that
\begin{equation}\label{Log-SY-L1}
\log S_Y(e^{i\theta}) =\log^+ S_Y(e^{i\theta}) -\log^-S_Y(e^{i\theta})\in L^1. 
\end{equation}
Here, we have used the notations $\log^+x=\max\{\log x,0\}$ and $\log^-x=\max\{-\log x,0\}$. By \eqref{eq:outer-factor}, there exists the unique normalized outer factor $G=\mathcal O_{S_Y}\in \mathcal{H}^2$ such that (see, e.g., \cite[Theorem 17.16]{rudin1987}, \cite[Theorem 2.4.1]{simon2005})
$$
S_Y(e^{i\theta})=|G(e^{i\theta})|^2
$$ 
almost everywhere, and moreover, 
\begin{equation}\label{SY-G0}
G(0)=\exp\left(\frac{1}{2}\int_{-\pi}^\pi\log S_Y(e^{i\theta})\frac{d\theta}{2\pi}\right)>0.
\end{equation}

Since $G$ is outer, it has no zeros in $\mathbb D$, so $\phi=F/G$ is analytic in $\mathbb D$. We first show that
\begin{equation}\label{eq:phi-multiplier}
\phi=\frac{F}{G}\in \mathcal{N}^+.
\end{equation}
To see this, define
$$
\omega(e^{i\theta})\triangleq\frac1{1+|F(e^{i\theta})|+|G(e^{i\theta})|}.
$$
Clearly, $\omega$ is a $\mu$-essentially bounded function, i.e.,
\begin{equation}\label{eq:ess-sup-o}
\operatorname*{ess\,sup}_{\theta\in[-\pi,\pi]}|\omega(e^{i\theta})|\le 1.
\end{equation}
Since $0\le-\log \omega\le|F|+|G|\in L^1$, there is a normalized outer function
\begin{equation}\label{eq:outer-fator-M}
M(z)=\exp\left\{\int_{-\pi}^\pi\frac{e^{i\theta}+z}{e^{i\theta}-z} \log \omega(e^{i\theta})\frac{d\theta}{2\pi}\right\},\quad z\in\mathbb{D}
\end{equation}
such that $|M(e^{i\theta})|=\omega(e^{i\theta})$ almost everywhere. By \cite[Theorem 17.16(a)]{rudin1987}, it holds that $\log |M|$ is the Poisson integral of $\log \omega$, i.e.,
\begin{equation}\label{eq:outer-fator-M-2}
\log |M(re^{i\varphi})|=\int_{-\pi}^\pi P_{r}(\varphi-\theta)\log\omega(e^{i\theta})\frac{d\theta}{2\pi},
\end{equation}
where $P_r(t)$ is the Poisson kernel such that (see, e.g., \cite[Section 11.5]{rudin1987})
\begin{equation}\label{eq:outer-fator-M-Poisson-kernel}
P_r(t)>0,\qquad \int_{-\pi}^\pi P_r(t)\frac{dt}{2\pi}=1, \quad 0\le r<1.
\end{equation}
Since $\log \omega \le 0$, it follows from \eqref{eq:outer-fator-M-2}--\eqref{eq:outer-fator-M-Poisson-kernel} that
\begin{equation}\label{eq:outer-fator-M-3}
\log|M(z)| \le 0.
\end{equation}
Thus, $M\in\mathcal{H}^\infty$ and $\| M\|_\infty\le1$. Set $a=MF$ and $b=MG$. Since $F,G\in \mathcal{H}^2$, we have $a,b\in \mathcal{H}^2\subset\mathcal{N}^+$ and $|a(e^{i\theta})|,|b(e^{i\theta})|\le1$ almost everywhere, which, together with the maximum principle of Smirnov (see also \cite[Lemma 2.3]{katsnelson1997theory}) implies that $a,b\in \mathcal{H}^\infty$. It then follows from \eqref{SY-L1}, \eqref{Log-SY-L1} and \eqref{eq:ess-sup-o} that
$$
\log \left(\omega(e^{i\theta})\sqrt{S_Y(e^{i\theta})}\right)=\log \omega(e^{i\theta}) +\log \sqrt{S_Y(e^{i\theta})}\in L^1(\mathbb T),\quad \omega(e^{i\theta})\sqrt{S_Y(e^{i\theta})}\in L^2(\mathbb T).
$$
Therefore, $b(z)\in \mathcal{H}^2$ is the outer function given by
$$
b(z)=M(z)G(z)=\exp\left\{\int_{-\pi}^\pi\frac{e^{i\theta}+z}{e^{i\theta}-z} \log\left( \omega(e^{i\theta})\sqrt{S_Y(e^{i\theta})}\right)\frac{d\theta}{2\pi}\right\}.
$$
Consequently,
\begin{equation*}
 \phi=\frac FG=\frac{MF}{MG}=\frac ab\in \mathcal{N}^+,
\end{equation*}
which establishes \eqref{eq:phi-multiplier}. Moreover, by \eqref{SY-F} and $|G|^2=S_Y$, we have
$$
|\phi|^2=\left |\frac{F}{G}\right|^2=\frac{|F|^2}{S_Y}\le 1, \qquad \text{a.e. on}  ~\mathbb T.
$$
Applying the maximum principle of Smirnov (see also \cite[Lemma 2.3]{katsnelson1997theory}) to $\phi$ now, we derive that $\phi \in \mathcal{H}^\infty$ and $\|\phi \|_\infty\le 1$. Thus, $\phi$ is a Schur function.

Since $S_Z$ is even, its normalized outer factor $H$ has real Taylor coefficients. Since $B$ has real Taylor coefficients and $S_V$ is even, $S_Y=S_V+|1+B|^2S_Z$ is also even. Hence, its normalized outer factor $G$ has real Taylor coefficients. Consequently, $H,G$ have real Taylor coefficients and it then follows algebraically that $F$ and $\phi$ have real Taylor coefficients as well.


To prove \eqref{eq:schur-defect}, we have 
\begin{equation}
\begin{aligned}
S_V(e^{i\theta}) &= S_Y(e^{i\theta})-|1+B(e^{i\theta})|^2S_Z(e^{i\theta})\\
&=S_Y(e^{i\theta}) - |1+B(e^{i\theta})|^2|H(e^{i\theta})|^2\\
&=|G(e^{i\theta})|^2 - |F(e^{i\theta})|^2\\
&=|G(e^{i\theta})|^2(1-|\phi (e^{i\theta})|^2),
\end{aligned}
\end{equation}
which proves \eqref{eq:schur-defect}. Note that 
$$
H(0)=\exp\left(\frac{1}{2}\int_{-\pi}^\pi\log S_Z(e^{i\theta})\frac{d\theta}{2\pi}\right)>0,
$$
which, together with \eqref{SY-G0}, immediately implies that
\begin{equation}\label{eq:rate-G-H}
\begin{aligned}
\mathcal{R}(S_V,B)=\frac{1}{2}\int_{-\pi}^\pi\log \frac{S_Y(e^{i\theta})}{S_Z(e^{i\theta})}\frac{d\theta}{2\pi}= \log\frac{G(0)}{H(0)},
\end{aligned}
\end{equation}
and consequently,
\begin{equation}\label{eq:phi-0}
\begin{aligned}
\phi(0)=\frac{F(0)}{G(0)}=\frac{H(0)}{G(0)}\in (0,1].
\end{aligned}
\end{equation}
As a result, equation \eqref{eq:rate-constant} follows from \eqref{eq:rate-G-H} and \eqref{eq:phi-0}. Finally,
\begin{equation*}
\begin{aligned}
\mathcal{P}(S_V,B)&=\int_{-\pi}^\pi (S_V(e^{i\theta})+|B(e^{i\theta})|^2S_Z(e^{i\theta}))\frac{d\theta}{2\pi}\\
&=\int_{-\pi}^\pi S_V(e^{i\theta})\frac{d\theta}{2\pi} +\int_{-\pi}^\pi |(BH)(e^{i\theta})|^2\frac{d\theta}{2\pi}\\
&=\int_{-\pi}^\pi \left(|G(e^{i\theta})|^2-|F(e^{i\theta})|^2 \right)\frac{d\theta}{2\pi}+ \int_{-\pi}^\pi |F(e^{i\theta})-H(e^{i\theta})|^2\frac{d\theta}{2\pi}\\
&=\| G\|_2^2-\| F\|_2^2 + \|F-H \|^2_2\\
&=\| G\|_2^2 + \|H\|^2_2 -2\mathrm{Re}\left\langle \phi G, H \right\rangle_{\mathcal{H}^2},
\end{aligned}
\end{equation*}
which proves \eqref{eq:power-affine}.
\end{proof}

Identity \eqref{eq:schur-defect} shows that $S_V = 0$ almost everywhere on $\mathbb T$ if and only if $\phi$ is inner. This observation motivates approximating $\phi$ by rational inner functions, namely, finite Blaschke products, that preserve $\phi(0)$ and the inner-product term in \eqref{eq:power-affine}.


\subsection{Finite Blaschke Approximation and Power Convergence}

\begin{pr}\label{pr2-3:label}
Under the hypotheses and notation of Lemma~\ref{lamma-SY:label}, there exists a sequence of finite Blaschke products $\{\Theta_n\}_{n\ge1}$ with real Taylor coefficients such that
\begin{equation}\label{inner-approx-Theta-n}
\Theta_n(0)=\phi(0),\qquad \lim_{n\to\infty}\Theta_n= \phi
\end{equation}
locally uniformly in $\mathbb{D}$ and weak-star in $\mathcal{H}^\infty$. Moreover, setting $Q_n\triangleq\Theta_nG$, we have, for every $n\ge1$, 
\begin{equation}\label{pr-inner-approx-Qn-H}
|Q_n|^2=S_Y\quad \text{a.e. on~}\mathbb T, \qquad Q_n(0)=H(0),
\end{equation}
and 
\begin{equation}\label{pr-inner-approx-limit}
\lim_{n\to\infty}\|Q_n-H \|_2^2=\mathcal{P}(S_V,B).
\end{equation}
\end{pr}

\begin{proof}
If $\phi$ is a finite Blaschke product, just take $\Theta_n=\phi$ for every $n$. Otherwise, by~\cite[Theorem 4.1.1 and its proof]{garcia2018finite}, there is a sequence of finite Blaschke products~$\Theta_n$ such that $\Theta_n(0)=\phi(0)$ for every $n$ and $\Theta_n$ uniformly converges to $\phi$ on compact subsets of $\mathbb{D}$, or equivalently, locally uniformly in $\mathbb{D}$ (see also~\cite{dijksma2005generalized}). To see that each $\Theta_n$ has real Taylor coefficients, observe that whenever $f$ has real Taylor coefficients, $w=f(0)$ is real. Thus, the initial factor $-b_w$ has real Taylor coefficients, and both the reduction $f(z)\mapsto -b_w(f(z))/z$ and the reconstruction $u(z)\mapsto -b_w(zu(z))$ preserve this property. Since $\phi$ has real Taylor coefficients, the conclusion follows by induction. Since $\Theta_n$ converges to $\phi$ at every point of $\mathbb{D}$ and $\|\Theta_n \|_\infty=1$ for every $n$, it follows from \cite[Lemma 1]{sarason1966weak} that $\{\Theta_n\}$ converges weak-star to the function $\phi$ in $\mathcal{H}^\infty$. This proves \eqref{inner-approx-Theta-n}. 

To prove \eqref{pr-inner-approx-Qn-H}, since $|\Theta_n |=1$ almost everywhere on $\mathbb T$, we have
$$
|Q_n(e^{i\theta})|^2=|\Theta_n(e^{i\theta})|^2|G(e^{i\theta})|^2= |G(e^{i\theta})|^2=S_Y(e^{i\theta})
$$
almost everywhere and 
$$
Q_n(0)=\Theta_n(0)G(0)=\phi(0)G(0)=F(0)=H(0),
$$
as desired.

We next prove \eqref{pr-inner-approx-limit}. Since $\{\Theta_n\}$ converges weak-star to the function $\phi$ in $\mathcal{H}^\infty$, it holds that for every $\omega(e^{i\theta})\in L^1(\mathbb{T})$
\begin{equation}\label{Theta_n-phi}
\lim_{n\to\infty}\int_{-\pi}^\pi\left(\Theta_n(e^{i\theta}) - \phi(e^{i\theta})\right)\omega(e^{i\theta})\frac{d\theta}{2\pi}= 0.
\end{equation}
For $h\in \mathcal{H}^2$, by the Cauchy--Schwarz inequality, we have 
$$
\int_{-\pi}^\pi \left|G(e^{i\theta})\overline{h(e^{i\theta})}\right|\frac{d\theta}{2\pi}\le \|G \|_2 \| h\|_2<\infty.
$$ 
Then, taking $\omega=G\overline h$ in \eqref{Theta_n-phi}, we obtain
\begin{equation}\label{eq-weak-H2}
\begin{aligned}
\lim_{n\to\infty}\left\langle Q_n-\phi G, h \right\rangle_{\mathcal{H}^2} &= \lim_{n\to\infty}\int_{-\pi}^\pi (\Theta_n(e^{i\theta}) - \phi(e^{i\theta}))G(e^{i\theta})\overline{h(e^{i\theta})}\frac{d\theta}{2\pi}\\
&=0.
\end{aligned}
\end{equation}
Since $h$ is arbitrary, \eqref{eq-weak-H2} proves that $Q_n=\Theta_n G$ weakly converges to $\phi G=F$ in $\mathcal{H}^2$. Furthermore, since $\Theta_n$ is inner, it holds that
$$
\|Q_n\|_2^2=\int_{-\pi}^\pi|\Theta_n(e^{i\theta}) G(e^{i\theta})|^2\frac{d\theta}{2\pi}=\|G \|_2^2.
$$
This, combined with \eqref{eq-weak-H2} and Lemma~\ref{lamma-SY:label}, immediately establishes
\begin{equation}
\begin{aligned}
    \lim_{n\to\infty}\|Q_n-H \|_2^2&=\lim_{n\to\infty}\left(\|Q_n \|_2^2 +\| H\|_2^2 - 2\mathrm{Re}\left\langle Q_n, H \right\rangle_{\mathcal{H}^2}\right)\\
    &=\|G \|_2^2+\| H\|_2^2- 2\lim_{n\to\infty}\mathrm{Re}\left\langle \Theta_n G, H \right\rangle_{\mathcal{H}^2}\\
    &=\|G \|_2^2+\| H\|_2^2 - 2\mathrm{Re}\left\langle \phi G, H \right\rangle_{\mathcal{H}^2}\\
    &=\mathcal{P}(S_V,B),
\end{aligned}
\end{equation}
which proves \eqref{pr-inner-approx-limit}. This completes the proof.
\end{proof}

Proposition~\ref{pr2-3:label} reproduces $S_Y$ at the level of the spectral factors $Q_n=\Theta_n G$. To obtain a corresponding filter-only pair with $(1+B_n)H=Q_n$, one would set $B_n=Q_n/H-1$; however, under \eqref{eq:PW-cond-All} alone, the induced ratio $Q_n/H-1$ need not belong to $z\mathcal{H}^2$. To overcome this obstacle, the next proposition uses Beurling's theorem \eqref{eq:outer-density} to construct real strictly causal polynomial filters with arbitrarily small power excess and rate loss.


\subsection{Strictly Causal Polynomial Approximation and Proof of Theorem~\ref{thm:main}}\label{sub:polynomial_approximation_and_proof_of_theorem_ref_thm_main}

\begin{pr}\label{theorem:pw-counterpart}
 Suppose that the Paley--Wiener condition \eqref{eq:PW-cond-All} holds, and let $(S_V,B)$ be admissible with input power $q\triangleq \mathcal{P}(S_V,B)$ and rate $R\triangleq \mathcal{R}(S_V,B)$. Then for every $\varepsilon,\delta>0$, there exists a real strictly causal polynomial $\widehat{B}$ such that
\begin{equation}\label{lemma:pw-two-bounds}
 \mathcal{P}(0,\widehat{B})\le q+\varepsilon, \qquad \mathcal{R}(0,\widehat{B})\ge R-\delta.
\end{equation}
\end{pr}

The proof of Proposition~\ref{theorem:pw-counterpart} is deferred to Appendix~\ref{appendix:Paley-Wiener Condition}. Because the polynomial filter $\widehat B$ in~\eqref{lemma:pw-two-bounds} may exceed the power constraint $q$ by up to $\varepsilon$, we prove Theorem~\ref{thm:main} by first scaling down the power of a given feasible pair and then applying Proposition~\ref{theorem:pw-counterpart}.


\begin{proof}[Proof of Theorem~\ref{thm:main}]
By \eqref{ine-C0-CFB}, it suffices to prove that $C_0(P)\ge C_\mathrm{FB}(P)$. To achieve this, fix a pair $(S_V,B)$ feasible at power $P$, and write $q=\mathcal P(S_V,B)$ and $R=R(S_V,B)$. For $0<t<1$, define
\begin{equation*}\label{eq:backoff}
 B_t=tB,\qquad S_{V,t}=tS_V+t(1-t)|B|^2S_Z.
\end{equation*}
Clearly, $(S_{V,t},B_t)$ is also an admissible pair. Then,
\begin{equation*}\label{eq:backoffidentities}
 \mathcal{P}(S_{V,t},B_t)=tq,\qquad
 S_{Y,t}=tS_Y+(1-t)S_Z.
\end{equation*}
By concavity of the logarithm,
\begin{equation*}
 \mathcal{R}(S_{V,t},B_t) =\frac12\int_\mathbb{T}\log\left(t\frac{S_Y}{S_Z}+1-t\right)d\mu \ge tR.
\end{equation*}
Applying Proposition~\ref{theorem:pw-counterpart} to $(S_{V,t},B_t)$ with $\varepsilon=(1-t)P/2$ and any $\eta>0$, we obtain a real strictly causal polynomial $\widehat B_t$ satisfying
$$
 \mathcal{P}(0,\widehat B_t)\le tq+\frac{(1-t)P}{2}\le\frac{1+t}{2}P<P
$$
and 
$$
\mathcal{R}(0,\widehat B_t)\ge \mathcal{R}(S_{V,t},B_t)-\eta\ge tR-\eta.
$$
Thus, $\widehat B_t$ is feasible at power $P$ and by~\eqref{eq:C-zero}, it holds that
$$
C_0(P)\ge tR-\eta.
$$
Letting $\eta\downarrow 0$ and then $t\uparrow1$, we obtain
$$
C_0(P)\ge R.
$$
Taking the supremum over all pairs $(S_V,B)$ feasible at power $P$ and using Kim's characterization \eqref{eq:Kim}, we establish
$$
C_0(P)\ge C_\mathrm{FB}(P),
$$
which, together with \eqref{ine-C0-CFB}, immediately proves \eqref{eq:main}, as desired. Each constructed filter $\widehat{B}_t$ is a real strictly causal polynomial. Taking the same suprema and limits with this restriction therefore proves the polynomial-sufficiency assertion as well.
\end{proof}

\begin{remark}\label{rem:scope}
Theorem~\ref{thm:main} establishes the filter-only capacity formula stated in \cite[Theorem 4.6]{kim2010feedback} under the general Paley--Wiener condition \eqref{eq:PW-cond-All}, thereby resolving the proof gap identified in~\cite{derpich2022comments} without requiring attainment of either supremum. Notably, the equality $C_\mathrm{FB}(P)=C_0(P)$ can hold even when neither supremum is attained. Indeed, if $S_Z$ is rational and its normalized outer factor $H$ vanishes at some point on $\mathbb T$, then Theorem~\ref{thm:optimizer-exist-rational} implies that no feasible pair attains \eqref{eq:Kim}; consequently, no feasible filter $B$ can attain \eqref{eq:C-zero}, since \eqref{eq:main} would otherwise imply that $(0,B)$ attains \eqref{eq:Kim}, a contradiction.
\end{remark}

\section{Exact Preservation of the Output Spectrum}\label{sec:exact_output_spectrum_preservation}
While Theorem~\ref{thm:main} addresses capacity via approximation, we now fix an arbitrary admissible pair $(S_V,B)$ and ask when filter-only pairs can reproduce its output spectrum $S_Y$ exactly. Under the Paley--Wiener condition \eqref{eq:PW-cond-All}, Section~\ref{sub:asymptotic_power_preservation_for_general_spectra} characterizes when a sequence of such filters achieves asymptotic power convergence to $\mathcal{P}(S_V,B)$. Section~\ref{sub:exact_power_preservation_for_rational_spectra} then shows that for rational noise spectra, the same integrability condition suffices for a single filter to preserve both $S_Y$ and the input power exactly.

\subsection{Asymptotic Power Preservation for General Spectra}\label{sub:asymptotic_power_preservation_for_general_spectra}
\begin{thm}\label{app-cor:ess-inf-cod}
Suppose that the Paley--Wiener condition \eqref{eq:PW-cond-All} holds, and let $(S_V,B)$ be an admissible pair with output spectrum $S_Y$, input power $q\triangleq \mathcal{P}(S_V,B)$, and rate $R\triangleq \mathcal{R}(S_V,B)$. Then there exists a sequence of real strictly causal filters $\{B_n\}_{n\ge1}\subset z\mathcal{H}^2$ such that
\begin{equation}
 S_Y=S_V+|1+B|^2S_Z=|1+B_n|^2S_Z\quad a.e.~\text{on~}\mathbb{T}~\mathit{for~every}~n\ge1\label{thm:approximation-kim-4-4-SY}
\end{equation}
and 
\begin{equation}\label{thm:approximation-kim-4-4-power}
\lim_{n\to\infty}\mathcal{P}(0,B_n)=\lim_{n\to\infty}\int_\mathbb{T}|B_n|^2S_Z\,d\mu= q
\end{equation}
if and only if
\begin{equation}\label{thm:approximation-kim-4-4-cond}
\frac{S_V}{S_Z}\in L^1(\mathbb{T}).
\end{equation}
In particular, when \eqref{thm:approximation-kim-4-4-cond} holds, for every $\varepsilon>0$, there exists a real strictly causal filter $B_n\in z\mathcal{H}^2$ satisfying $\mathcal{P}(0,B_n)\le q+\varepsilon$ and $\mathcal{R}(0,B_n)=R$.
\end{thm}

\begin{proof}
For necessity, suppose that there are real strictly causal filters $B_n\in z\mathcal{H}^2$ satisfying \eqref{thm:approximation-kim-4-4-SY}--\eqref{thm:approximation-kim-4-4-power}. Taking $n=1$ in \eqref{thm:approximation-kim-4-4-SY}, we obtain
$$
\frac{S_V}{S_Z}=|1+B_1|^2-|1+B|^2\in L^1(\mathbb{T}).
$$
Conversely, assume that $S_V/S_Z\in L^1(\mathbb{T})$. Let $H$ and $G$ be the normalized outer factors of $S_Z$ and $S_Y$, respectively, and let $\phi=(1+B)H/G$, as in Lemma~\ref{lamma-SY:label}. Choose the real finite Blaschke products $\Theta_n$ of Proposition~\ref{pr2-3:label}, let $Q_n=\Theta_nG$, and set
\begin{equation}\label{proof:eq-Bn-construction}
B_n=\frac{Q_n}{H}-1=\frac{\Theta_n G}{H}-1.
\end{equation}
Since $B\in \mathcal{H}^2$, it holds that
$$
W\triangleq\frac{S_Y}{S_Z}=\frac{S_V}{S_Z}+|1+B|^2\in L^1(\mathbb{T}).
$$
It then follows from \eqref{Log-SY-L1} in Lemma~\ref{lamma-SY:label} and \eqref{eq:PW-cond-All} that $\log W\in L^1(\mathbb{T})$. Thus, there exists a unique normalized outer spectral factor $K=O_W\in\mathcal{H}^2$ of $W$ such that for every $z\in\mathbb{D}$,
\begin{equation}\label{proof:eq-K-H2}
  \begin{aligned}
    K(z) &= \exp \left\{ \frac{1}{2} \int_{-\pi}^\pi \frac{e^{i\theta} + z}{e^{i\theta}- z} \log \frac{S_Y(e^{i\theta})}{S_Z(e^{i\theta})} \, \frac{d\theta}{2\pi} \right\} \\
&= \frac{G(z)}{H(z)}, \qquad K \in \mathcal{H}^2.
  \end{aligned}
\end{equation}
Since $\Theta_n\in \mathcal{H}^\infty$, equations \eqref{proof:eq-Bn-construction} and \eqref{proof:eq-K-H2} imply that $B_n\in\mathcal{H}^2$. Moreover, $H,G,\Theta_n$ have real Taylor coefficients, and \eqref{pr-inner-approx-Qn-H} gives $B_n(0)=0$. Thus, $B_n\in z\mathcal{H}^2$ is a real strictly causal filter. Furthermore, since $|\Theta_n|=1$ almost everywhere on $\mathbb{T}$,
$$
|1+B_n|^2S_Z=|Q_n|^2=|G|^2=S_Y,\quad \mathrm{a.e.,}
$$
which proves \eqref{thm:approximation-kim-4-4-SY} and, by \eqref{eq:rate}, yields $\mathcal{R}(0,B_n)=R$ for every $n$. Finally, it follows from~\eqref{proof:eq-Bn-construction} that $B_nH=Q_n-H$. Hence, by \eqref{pr-inner-approx-limit}, we have
\begin{equation}\label{B_n-power}
\lim_{n\to\infty}\mathcal{P}(0,B_n)=\lim_{n\to\infty}\|B_nH \|_2^2=\lim_{n\to\infty}\|Q_n-H\|_2^2=\mathcal{P}(S_V,B)=q.
\end{equation}
This proves \eqref{thm:approximation-kim-4-4-power}.
\end{proof}

\begin{remark}
\label{rem:integrability}
If $S_Z$ is bounded away from zero, then every admissible pair $(S_V, B)$ satisfies~\eqref{thm:approximation-kim-4-4-cond}, since
$$
\int_{\mathbb T}\frac{S_V}{S_Z}\,d\mu\leq \frac{1}{\operatorname*{ess\,inf}_{\mathbb T}S_Z} \int_{\mathbb T}S_V\,d\mu<\infty.
$$
Hence, by Theorem~\ref{app-cor:ess-inf-cod}, every admissible pair $(S_V, B)$ admits a sequence $\{B_n\}_{n \ge 1} \subset z\mathcal{H}^2$  satisfying~\eqref{thm:approximation-kim-4-4-SY}--\eqref{thm:approximation-kim-4-4-power}. On the other hand, condition~\eqref{thm:approximation-kim-4-4-cond} can fail under the Paley--Wiener condition~\eqref{eq:PW-cond-All}. For example, take $S_Z(e^{i\theta}) = |1 - e^{i\theta}|^2$, $S_V \equiv 1$, and $B \equiv 0$. The pair $(1, 0)$ is admissible, and since $S_Z(e^{i\theta})=4\sin^2(\theta/2)$, we have $\log S_Z \in L^1(\mathbb{T})$, whereas $S_V/S_Z \notin L^1(\mathbb{T})$. By Theorem~\ref{app-cor:ess-inf-cod}, no real strictly causal filter in $z\mathcal{H}^2$ can reproduce the corresponding output spectrum exactly.
\end{remark}

\subsection{Exact Power Preservation for Rational Spectra}\label{sub:exact_power_preservation_for_rational_spectra}
We now specialize Theorem~\ref{app-cor:ess-inf-cod} to rational noise spectra. In the notation of Section~\ref{sec:attainment_and_uniqueness_in_kim_s_variational_problem}, a stationary Gaussian noise process with normalized outer factor $H=N_H/D_H$ admits the finite-order autoregressive moving-average (ARMA) representation
$$
D_H(L)Z_i=N_H(L)W_i,
$$
where $L$ is the lag operator satisfying $LZ_i=Z_{i-1}$ and $\{W_i\}$ is white Gaussian noise with unit variance. The autoregressive and moving-average orders are $\deg D_H$ and $\deg N_H$, respectively. Note that the moving-average polynomial $N_H$ may have zeros on $\mathbb T$, whereas Kim \cite[Section V, equation (40)]{kim2010feedback} restricted attention to the subclass where $N_H$ is nonvanishing on $\mathbb T$ (equivalently, $S_Z$ is bounded away from zero).


For a nonzero polynomial $r(z)=\sum_{j=0}^{\ell}r_jz^j$ of degree $\ell$, we define the {\em reversed polynomial} of $r(z)$ by
$$
r^\sharp(z) \triangleq\sum_{j=0}^\ell\overline{r}_{\ell-j}z^j;
$$
see \cite[(1.1.6)--(1.1.7)]{simon2005}. Equivalently, $r^\sharp(z)=z^\ell\overline{r(1/\overline z)}$ for $z\ne0$. In particular, $r^\sharp(z)=z^\ell\overline{r(z)}$ on $z\in\mathbb T$. 

The next lemma provides the finite interpolation needed to preserve the inner product in \eqref{eq:power-affine} and hence the input power. Its proof is given in Appendix~\ref{appendix:lemma-Schur-completion}.

\begin{lem}\label{lemma-Schur-completion}
Let $\varphi$ be a Schur function with real Taylor coefficients, and let $\alpha_1,\ldots,\alpha_n\in\mathbb{D}$ be distinct points with positive integer multiplicities $m_1,\ldots,m_n$, respectively. Assume that for each $i\in\{1,\ldots,n\}$ there exists $j\in\{1,\ldots,n\}$ such that $\alpha_j=\overline{\alpha_i}$ and $m_j=m_i$. Then there exists a finite Blaschke product $\Theta$ with real Taylor coefficients such that
$$
\deg \Theta \le \sum_{i=1}^nm_i,\qquad \Theta^{(\ell)}(\alpha_i)=\varphi^{(\ell)}(\alpha_i),\quad 0\le \ell<m_i, \quad 1\le i\le n.
$$
Moreover, for every finite Blaschke product $\beta$ whose zeros in $\mathbb D$ are precisely $\alpha_1,\ldots,\alpha_n$ with multiplicities $m_1,\ldots,m_n$, respectively, we have $(\Theta-\varphi)/\beta\in\mathcal{H}^\infty$.
\end{lem}

Applying Lemma~\ref{lemma-Schur-completion} at the zeros of a finite Blaschke product determined by the rational outer factor $H$, we obtain the following exact preservation result.

\begin{thm}\label{thm:exact-removal-rational}
Suppose that $S_Z$ is rational. Let $(S_V,B)$ be an admissible pair with output spectrum $S_Y$ and input power $q\triangleq \mathcal{P}(S_V,B)$. Then there exists a real strictly causal filter $\widetilde B\in z\mathcal{H}^2$ satisfying
\begin{equation}\label{thm:B-ast}
\begin{aligned}
S_Y = S_V+|1+B|^2S_Z=|1+\widetilde B|^2S_Z  \quad\text{a.e. on }\mathbb T,\qquad \mathcal{P}(0,\widetilde B)=q
\end{aligned}
\end{equation}
if and only if $S_V/S_Z\in L^1(\mathbb T)$.
\end{thm}

\begin{proof}
For necessity, suppose that there exists a real strictly causal filter $\widetilde B\in z\mathcal{H}^2$ satisfying \eqref{thm:B-ast}. Then, we have
$$
\frac{S_V}{S_Z}=|1+\widetilde B|^2-|1+B|^2\in L^1(\mathbb{T}).
$$

Conversely, for sufficiency, assume that $S_V/S_Z\in L^1(\mathbb T)$. Let $G=\mathcal O_{S_Y}$ and $\phi=(1+B)H/G$ be as in Lemma~\ref{lamma-SY:label}. Let $W\triangleq S_Y/S_Z$. Then, 
$$
W=\frac{S_V}{S_Z}+|1+B|^2\in L^1(\mathbb T),\qquad \log W=\log S_Y -\log S_Z\in L^1(\mathbb T).
$$
By \eqref{eq:outer-factor}, there exists a normalized outer spectral factor $\mathcal O_W$ such that
\begin{equation}\label{pf-thm-3-4-O-W}
\mathcal{O}_W=\frac{G}{H}\in\mathcal H^2.
\end{equation}

Let $m=\deg N_H$, $n=\deg D_H$, and $d=\max\{m,n\}$. Define
$$
\beta(z)\triangleq z^{d+1-n}\frac{D_H^\sharp(z)}{D_H(z)}.
$$
Since $D_H(0)=1$ and $D_H$ has no zeros in $\overline{\mathbb D}$, we may write
$$
D_H(z)=\prod_{j=1}^n(1-\overline{a_j}z), \qquad a_j\in\mathbb D,
$$
where the zeros of $D_H(z)$ are $1/\overline{a_j}$, counted with multiplicity. Thus, $D_H^\sharp(z)=\prod_{j=1}^n(z-a_j)$, and
$$
 \beta(z)=z^{d+1-n}\prod_{j=1}^n b_{a_j}(z).
$$
Thus, $\beta(z)$ is a finite Blaschke product of degree $d+1$. Also, $\beta(z)$ has real Taylor coefficients because $D_H$ has real coefficients, and $\beta(0)=0$ because $d+1-n\ge1$. Moreover, on $\mathbb T$, it follows from $\overline{N_H(z)}=z^{-m}N_H^\sharp(z)$ and $\overline{D_H(z)}=z^{-n}D_H^\sharp(z)$ that
\begin{equation}\label{pf-thm-3-4-beta-H}
\beta(z)\overline{H(z)} = z^{d+1-n}\frac{D_H^\sharp(z)}{D_H(z)} z^{n-m}\frac{N_H^\sharp(z)}{D_H^\sharp(z)} =  z^{d+1-m}\frac{N_H^\sharp(z)}{D_H(z)}.
\end{equation}
Since $D_H(z)$ has no zeros in $\overline{\mathbb D}$ and $d+1-m\ge1$, $z^{d+1-m}N_H^\sharp(z)/D_H(z)\in z\mathcal H^\infty$ in~\eqref{pf-thm-3-4-beta-H}. Therefore, for every $f\in\mathcal H^2$,
\begin{equation}\label{pf-thm-3-4-betaf-H}
\langle\beta f,H\rangle_{\mathcal H^2} = \int_{\mathbb T} f(z)z^{d+1-m}\frac{N_H^\sharp(z)}{D_H(z)}\,d\mu(z) =0.
\end{equation}
Indeed, the integrand in \eqref{pf-thm-3-4-betaf-H} is the boundary function of $z^{d+1-m}fN_H^\sharp/D_H\in z\mathcal H^2$, since $f\in\mathcal H^2$ and $z^{d+1-m}N_H^\sharp/D_H\in z\mathcal H^\infty$.

Since $\beta$ has real Taylor coefficients, the nonreal zeros of $\beta$ occur in conjugate pairs of equal multiplicity. Applying Lemma~\ref{lemma-Schur-completion} to $\phi$ at all zeros of $\beta(z)$, counted with their multiplicities, we obtain a finite Blaschke product $\Theta$ with real Taylor coefficients such that
$$
\deg \Theta \le d+1,\qquad \Theta-\phi =\beta \psi \quad \mbox{for some }\psi\in\mathcal H^\infty.
$$
Evaluating $\Theta-\phi=\beta\psi$ at zero and using $\beta(0)=0$ and $F(0)=H(0)$, we obtain
\begin{equation}\label{pf-thm-3-4-Theta-0}
\Theta(0)=\phi(0)=\frac{H(0)}{G(0)}.
\end{equation}
Since $\psi\in\mathcal H^\infty$ and $G\in\mathcal H^2$, we may take $f=\psi G$ in~\eqref{pf-thm-3-4-betaf-H} to obtain
\begin{equation}\label{pf-thm-3-4-inner-prod}
  \begin{aligned}
    \langle\Theta G-\phi G,H\rangle_{\mathcal H^2}=\langle(\Theta-\phi) G,H\rangle_{\mathcal H^2}=\langle \beta \psi G,H\rangle_{\mathcal H^2}=0.
  \end{aligned}
\end{equation}
Define 
$$
\widetilde B=\frac{\Theta G}{H}-1.
$$
It follows from \eqref{pf-thm-3-4-O-W} and \eqref{pf-thm-3-4-Theta-0} that $\widetilde B\in z\mathcal H^2$ with real Taylor coefficients. Moreover, $\widetilde B H=\Theta G-H\in\mathcal H^2$. Therefore,
\begin{align*}
\mathcal{P}(0, \widetilde B) &= \|\Theta G - H\|_2^2 \\
&= \|G\|_2^2 + \|H\|_2^2 - 2 \operatorname{Re}\langle \Theta G, H \rangle_{\mathcal H^2} \\
&= \|G\|_2^2 + \|H\|_2^2 - 2 \operatorname{Re}\langle \phi G, H \rangle_{H^2}\\
&=\mathcal P(S_V,B)= q,
\end{align*}
where the second equality follows from $|\Theta|=1$ almost everywhere on $\mathbb T$, the third from \eqref{pf-thm-3-4-inner-prod}, and the fourth from \eqref{eq:power-affine}. Finally, 
$$
|1+\widetilde B|^2S_Z=|\Theta G|^2=S_Y \quad \mbox{a.e. on }\mathbb T.
$$
This proves \eqref{thm:B-ast}.
\end{proof}

\begin{remark}
If $S_Z$ is rational and bounded away from zero, then $1/S_Z\in L^\infty(\mathbb T)$, so every admissible pair $(S_V,B)$ automatically satisfies $S_V/S_Z\in L^1(\mathbb T)$. By~Theorem~\ref{thm:exact-removal-rational}, every admissible pair, not merely a maximizing pair, admits a single filter-only replacement with the exact same output spectrum and input power.
\end{remark}

\section{Explicit Optimal Filter for a Class of AR(2) Gaussian Channels}\label{sec:sk2-special_channels}

Explicit capacity-achieving SK-type schemes are known for additive white Gaussian noise channels~\cite{schalkwijk1966coding} and for stationary $\mathrm{AR}(1)$ Gaussian channels~\cite[Section~V, equations~(48)--(49)]{kim2010feedback}. To illustrate the results of Section~\ref{sec:attainment_and_uniqueness_in_kim_s_variational_problem}, the following proposition explicitly identifies the unique optimal feedback filter for the feedback capacity of a class of stationary AR(2) Gaussian channels whose noise consists of two interleaved AR(1) processes.

\begin{pr}\label{prop:interleaved-ar1}
Consider the stationary AR(2) noise process $Z_i+\beta Z_{i-2}=W_i$, where $\beta\in\mathbb R$ satisfies $0<|\beta|<1$ and  $\{W_i\}$ is zero-mean unit-variance white Gaussian noise. For every $P>0$, let $s\triangleq\operatorname{sgn}(\beta)$ and let $x\in(0,1)$ be the unique solution to
\begin{equation}\label{eq:Section-5-x}
P x^2(1+|\beta|x)^2=1-x^2.
\end{equation}
Then $(0,B^\ast)$ is the unique maximizing pair of \eqref{eq:Kim} at power $P$, where
\begin{equation}\label{eq:Section-5-B-star}
B^{\ast}(z)=-s\,\frac{1-x^2}{x(1+|\beta|x)}\frac{z^2(1+\beta z^2)}{1-sxz^2},
\end{equation}
and it achieves $\mathcal{P}(0,B^\ast)=P$ and $\mathcal{R}(0,B^\ast)=C_\mathrm{FB}(P)=-\log x$.
\end{pr}

\begin{proof}
Note that the normalized outer factor of $S_Z$ is $H(z) = (1 + \beta z^2)^{-1}$, and
$$
S_Z(e^{i\theta}) = \frac{1}{|1 + \beta e^{2i\theta}|^2} \ge \frac{1}{(1 + |\beta|)^2} > 0.
$$
Hence, $S_Z$ is rational and bounded away from zero; since $\beta \ne 0$, $S_Z$ is also nonwhite. By Corollary~\ref{co:label-sec-2}, for every $P > 0$, the mixed variational problem~\eqref{eq:Kim} admits a unique maximizing pair, which has the filter-only form $(0, B^\ast)$. We now verify that this optimal filter is given by~\eqref{eq:Section-5-B-star}.

First, it is clear that~\eqref{eq:Section-5-x} has a unique solution $x \in (0,1)$. Moreover, $B^\ast$ has real Taylor coefficients and satisfies $|1 - s x z^2| \ge 1 - x > 0$ for all $z \in \overline{\mathbb{D}}$, so
\begin{equation}\label{eq:Section-5-B-star-quali}
B^{\ast}\in z^2\mathcal H^\infty\subset z\mathcal H^2.
\end{equation}
Multiplying~\eqref{eq:Section-5-B-star} by $H$ yields
$$
B^{\ast}(z)H(z) =-s\,\frac{1-x^2}{x(1+|\beta|x)}  \sum_{k=0}^{\infty}(sx)^k z^{2k+2}.
$$
By Parseval's identity and \eqref{eq:Section-5-x},
\begin{equation}\label{eq:section-5-power}
\mathcal P(0,B^{\ast}) =\|B^{\ast}H\|_2^2 =\frac{(1-x^2)^2}{x^2(1+|\beta|x)^2}\sum_{k=0}^{\infty}x^{2k}=\frac{1-x^2}{x^2(1+|\beta|x)^2}=P.
\end{equation}
By~\eqref{eq:Section-5-B-star-quali} and~\eqref{eq:section-5-power}, $B^\ast$ is a real strictly causal filter and $(0, B^\ast)$ is feasible at power $P$.

To evaluate the rate $\mathcal{R}(0, B^\ast)$, let $\zeta_1, \zeta_2$ be the two roots of $z^2 = s x$, both of which have modulus $\sqrt{x} < 1$. Using $\beta s = |\beta|$ in~\eqref{eq:Section-5-B-star}, we obtain
$$
B^{\ast}(\zeta_j) =-s\,\frac{1-x^2}{x(1+|\beta|x)} \frac{sx(1+|\beta|x)}{1-x^2}=-1,\qquad j=1,2.
$$
Thus, $1 + B^\ast$ is analytic in a neighborhood of $\overline{\mathbb{D}}$, equals $1$ at the origin, and vanishes at $\zeta_1$ and $\zeta_2$. By Jensen's formula,
\begin{equation}\label{eq:section-5-rate-jensen}
\mathcal R(0,B^{\ast}) =\int_{\mathbb T}\log|1+B^{\ast}|\,d\mu\ge \sum_{j=1}^2\log\frac{1}{|\zeta_j|}=-\log x.
\end{equation}
On the other hand, it follows from~\cite[Proposition~3.9, equations~(96)--(97)]{su2026secondorder} that
\begin{equation}\label{eq:section-5-Cfb}
C_{\mathrm{FB}}(P) = \log x_0,
\end{equation}
where $x_0 > 1$ is the unique solution to
$$
P = \frac{x_0^2 - 1}{(1 + |\beta|/x_0)^2}.
$$
Comparing with~\eqref{eq:Section-5-x} gives $x_0 = 1/x$. Since $(0, B^\ast)$ is feasible at power $P$, it follows from~\eqref{eq:Kim} and~\eqref{eq:section-5-Cfb} that
\begin{equation}\label{eq:section-5-rate-Cfb}
\mathcal{R}(0, B^\ast) \le C_{\mathrm{FB}}(P) = -\log x.
\end{equation}
Combining \eqref{eq:section-5-power}, \eqref{eq:section-5-rate-jensen} and \eqref{eq:section-5-rate-Cfb}, we obtain
\begin{equation}\label{eq:Section-5-B-star-power-rate}
\mathcal P(0,B^{\ast})=P,\qquad \mathcal R(0,B^{\ast})=C_{\mathrm{FB}}(P)=-\log x.
\end{equation}
By Corollary~\ref{co:label-sec-2}, $(0, B^\ast)$ is therefore the unique maximizing pair in~\eqref{eq:Kim}.
\end{proof}

Notice that the even- and odd-indexed noise subsequences are independent stationary $\mathrm{AR}(1)$ processes. Since $B^\ast(z)$ depends only on $z^2$, it applies the same feedback filter separately to the two subsequences, in agreement with the interleaved $\mathrm{SK}(2)$ construction in~\cite[Proposition~3.9]{su2026secondorder}.

\section{Conclusion}\label{sec:conclusion}
We have resolved the filter-only characterization of the feedback capacity for power-constrained additive stationary Gaussian channels under the general Paley--Wiener condition~\eqref{eq:PW-cond-All}. Using outer spectral factorization, finite Blaschke approximation, and Beurling's theorem, we proved that the mixed and filter-only capacity suprema coincide at every average power constraint $P>0$ and that real strictly causal polynomials suffice to approach capacity (Theorem~\ref{thm:main}), thereby establishing~\cite[Theorem~4.6]{kim2010feedback} without requiring attainment. We further characterized uniqueness, vanishing independent input spectra, and attainment in the mixed variational problem (Propositions~\ref{pr:nonwhite-optmimzer-existence-charaterziation} and~\ref{pr:nonwhite-rational-optimizer-removal}, Theorem~\ref{thm:optimizer-exist-rational}), established the integrability criterion $S_V/S_Z \in L^1(\mathbb{T})$ for exact output-spectrum preservation with asymptotic or exact power preservation (Theorems~\ref{app-cor:ess-inf-cod} and~\ref{thm:exact-removal-rational}), and explicitly identified the unique optimal feedback filter for the feedback capacity of a class of stationary AR(2) (interleaved AR(1)) Gaussian channels (Proposition~\ref{prop:interleaved-ar1}).

Several natural questions remain open for future work. First, for nonwhite \emph{nonrational} spectra that are bounded away from zero, a unique maximizer $(S_V^\ast, B^\ast)$ exists by Proposition~\ref{pr:nonwhite-optmimzer-existence-charaterziation}, and it remains open whether $S_V^\ast = 0$ almost everywhere on $\mathbb{T}$ in general, as well as whether exact single-filter preservation of both the output spectrum and input power extends beyond rational spectra. Second, it is of interest to extend the Hardy-space and inner-approximation framework developed here, via matrix-valued spectral factorization and Blaschke--Potapov products~\cite{katsnelson1997theory}, to stationary MIMO colored Gaussian channels with feedback. Finally, translating the strictly causal polynomial filter approximations of Theorem~\ref{thm:main} into explicit capacity-achieving coding schemes for general nonrational noise spectra (or rational spectra vanishing on $\mathbb{T}$) is an appealing direction for further investigation.

\section*{Appendices}
\appendix
\section{Proof of Proposition~\ref{theorem:pw-counterpart}}\label{appendix:Paley-Wiener Condition}
We use Proposition~\ref{pr2-3:label} to select a finite Blaschke product $\Theta$ with real Taylor coefficients. Since $H$ is outer and $\Theta G-H\in z\mathcal H^2$, we can choose real strictly causal polynomials $B_m$ such that
\[
    \lim_{m\to\infty}\|B_mH-(\Theta G-H)\|_2 = 0.
\]
We then use Rouch\'e's theorem to locate zeros of $1+B_m$ and Jensen's formula to obtain the required lower bound on the rate.

\begin{proof}
Let $H$ and $G$ be the normalized outer factors of $S_Z$ and $S_Y$, respectively, and let $\phi=(1+B)H/G$. It then follows from Lemma~\ref{lamma-SY:label} that  $H,G\in\mathcal H^2$ have real Taylor coefficients, $\phi$ is a Schur function, and
$$
R=\log\frac{G(0)}{H(0)}=-\log\phi(0)\ge 0.
$$
Clearly, if $R=0$, then $\widehat B=0$ satisfies \eqref{lemma:pw-two-bounds}. Therefore, throughout the remainder of the proof, we assume $R>0$ and fix $\varepsilon,\delta>0$.

It then follows from Proposition~\ref{pr2-3:label} that there is a finite Blaschke product $\Theta$ with real Taylor coefficients such that
\begin{equation}\label{eq:pf-PW-1}
 \Theta(0)=\frac{H(0)}{G(0)}=e^{-R}>0, \qquad \|\Theta G-H\|_2^2\le q+\frac{\varepsilon}{2}, 
\end{equation}
which implies that $D\triangleq\Theta G-H\in z\mathcal{H}^2$. Furthermore, let $K\triangleq D/z$ for $z\neq 0$ and $K(0)=D^\prime(0)$. Clearly, $K\in \mathcal{H}^2$ and $\|K\|_2=\| D\|_2$. Applying \eqref{eq:outer-density} to $F=K$ with $\varepsilon=1/m$, we obtain polynomials $p_m$ such that
\begin{equation}\label{eq:poly-approx-pw}
  \lim_{m\to\infty}\|p_m H-K \|_2=0.
\end{equation}
In particular, the polynomials $p_m$ may be chosen to have real coefficients. Indeed, for $f\in \mathcal{H}^2$, define 
$$
f^\#(z)\triangleq \overline{f(\bar{z})}.
$$ 
Clearly, $\|f^\#\|_2=\|f\|_2$. It then follows from Lemma~\ref{lamma-SY:label} and Proposition~\ref{pr2-3:label} that $H,G$ and $\Theta$ have real Taylor coefficients. Hence, $H^\#=H$ and $K^\#=K$, and therefore
$$
\left\lVert\frac{p_m+p_m^\#}{2}H-K\right\rVert_2 \le\frac{\|{p_mH-K}\|_2+\|(p_mH-K)^\#\|_2}{2} =\|p_mH-K\|_2.
$$
Thus, replacing $p_m$ by $(p_m+p_m^\#)/2$ preserves \eqref{eq:poly-approx-pw} and gives a polynomial with real coefficients. Let $B_m=zp_m$, which is a real strictly causal polynomial. Set
$$
q_m = \int_\mathbb{T}|B_m|^2S_Zd\mu.
$$
Clearly, $q_m<\infty$ since $B_m$ is bounded on $\mathbb{T}$ and $S_Z\in L^1$. Then, by \eqref{eq:poly-approx-pw}, we have
\begin{equation}\label{eq:pf-PW-2}
\begin{aligned}
\lim_{m\to\infty}\|B_m H-D\|_2&=\lim_{m\to\infty}\|zp_m H-zK\|_2\\
&=\lim_{m\to\infty}\|p_mH-K \|_2\\
&=0,
\end{aligned}
\end{equation}
and thus
\begin{equation}\label{eq:pf-PW-3}
\begin{aligned}
\lim_{m\to\infty}q_m=\lim_{m\to\infty}\|B_m H \|_2^2=\| D\|^2_2.
\end{aligned}
\end{equation}
Consequently, from \eqref{eq:pf-PW-1} to \eqref{eq:pf-PW-3}, we obtain that
\begin{equation}\label{eq:power-bound-pw-final}
q_m\le q+\varepsilon 
\end{equation}
for all sufficiently large $m$. 

We next establish the second inequality in \eqref{lemma:pw-two-bounds}. Let $U_m\triangleq1+B_m$ and $g\triangleq\Theta G/H$. Since the outer function $H$ has no zeros in $\mathbb{D}$, $g$ is analytic and $g(0)=1$. For $0<r<1$, let 
$$
c_r\triangleq\min_{|z|\le r}|H(z)|>0.
$$
Since $U_m-g =(B_mH-D)/H$, it holds that
\begin{equation}\label{eq:sup-converge-local}
\begin{aligned}
\sup_{|z|\le r} |U_m(z) -g(z)| &\le \frac{1}{c_r}\sup_{|z|\le r}|B_m(z)H(z)-D(z)|\\
&\le \frac{\| B_mH-D\|_2}{c_r\sqrt{1-r^2}},
\end{aligned}
\end{equation}
where the last inequality follows from the Cauchy--Schwarz inequality and \cite[Theorem 17.12]{rudin1987}. Consequently, combining \eqref{eq:pf-PW-2} and \eqref{eq:sup-converge-local}, we have derived that 
\begin{equation}\label{eq:local-converge}
\lim_{m\to\infty}(1+B_m) =\frac{\Theta G}{H}
\end{equation}
locally uniformly in $\mathbb{D}$.


Since both $G$ and $H$ are nonvanishing in $\mathbb{D}$, the zeros of $g$ are precisely the finitely many zeros of $\Theta$, with the same multiplicities. Write
$$
\Theta(z) = \zeta \prod_{j=1}^{s} \left( \frac{z - a_j}{1 - \overline{a_j} z} \right)^{m_j}, \qquad |\zeta| = 1,
$$
where the $a_j$ are distinct, $m_j\ge 1$ and $m_j$ is the multiplicity of the zero $a_j$. Since $0<\Theta(0)=e^{-R}<1$, we have $s\ge 1$, $0 < |a_j| < 1$ for every $j$ and 
\begin{equation}\label{eq:R-mj-aj}
\sum_{j=1}^{s} m_j \log \frac{1}{|a_j|} = -\log |\Theta(0)| = R.
\end{equation}
Choose $\rho_j>0$ such that the closed disks
\begin{equation}\label{eq:pw-def-Delta-circle}
\overline{\Delta_j} = \{ z : |z - a_j| \leq \rho_j \}
\end{equation}
are pairwise disjoint, $\overline{\Delta_j}\subset \mathbb{D}\setminus\{0\}$ and satisfy
\begin{equation}\label{eq:pw-cond-Delta-circle}
\sum_{j=1}^{s} m_j \log \frac{1}{|a_j| + \rho_j} > R - \frac{\delta}{2}.
\end{equation}
Such $\rho_j$ always exist since
\begin{align*}
R - \sum_{j=1}^{s} m_j \log \frac{1}{|a_j| + \rho_j} &= \sum_{j=1}^{s} m_j \log \left( 1 + \frac{\rho_j}{|a_j|} \right) \\
&\leq \sum_{j=1}^{s} m_j \frac{\rho_j}{|a_j|},
\end{align*}
where the equality follows from \eqref{eq:R-mj-aj}. Consequently, the only zero of $g$ in $\Delta_j$ is $a_j$, with multiplicity $m_j$; in particular, the function $g$ is nonzero on each circle $\partial \Delta_j$. Let 
$$
d_j \triangleq \min_{z \in \partial \Delta_j} |g(z)| > 0.
$$
Since $\Gamma\triangleq\bigcup_{j=1}^s\partial \Delta_j$ is compact in $\mathbb{D}$, by the local uniform convergence in \eqref{eq:local-converge}, there exists $M$ such that
$$
\sup_{z \in \Gamma} |U_m(z) - g(z)| < \frac{\min_j d_j}{2} \qquad \text{for all }m \geq M.
$$
Then, simultaneously for every $j$ and every $m \geq M$, we obtain
$$
|U_m(z) - g(z)| < \frac{\min_j d_j}{2} < d_j \leq |g(z)| \qquad \text{for}~z \in \partial \Delta_j.
$$
It then follows from Rouch\'e's theorem \cite[Theorem 10.43(b)]{rudin1987} that, for all sufficiently large $m$, $U_m$ has exactly $m_j$ zeros in the interior of $\Delta_j$, counted with multiplicity. Denote these zeros of $U_m$ by $\alpha_{j,\ell}^{(m)}$, $1 \leq \ell \leq m_j$. It therefore follows from \eqref{eq:pw-def-Delta-circle} that
\begin{equation}\label{eq:pw-ineq-Delta-circle}
0 < |\alpha_{j,\ell}^{(m)}| < |a_j| + \rho_j < 1. 
\end{equation}

Since $U_m$ is a polynomial and $U_m(0)=(1 + B_m)(0) = 1$, by Jensen's formula \cite[Theorem 15.18]{rudin1987}, it holds that
\begin{equation}\label{eq:pw-jensen-formula}
\int_{\mathbb{T}} \log |U_m| \, d\mu = \sum_{\substack{U_m(\alpha)=0 \\ |\alpha|<1}} \log \frac{1}{|\alpha|},
\end{equation}
where the zeros are counted with multiplicity. Possible zeros on $\mathbb{T}$ have integrable logarithmic singularities and contribute zero to this identity; see \cite[Lemma 15.17 and Theorem 15.18]{rudin1987}. Thus, every term on the right of \eqref{eq:pw-jensen-formula} is nonnegative. Combining \eqref{eq:pw-cond-Delta-circle}, \eqref{eq:pw-ineq-Delta-circle} and \eqref{eq:pw-jensen-formula}, we obtain
\begin{equation}\label{eq:rate-bound-pw-final}
\begin{aligned}
\int_{\mathbb{T}} \log |1 + B_m| \, d\mu &\geq \sum_{j=1}^{s} \sum_{\ell=1}^{m_j} \log \frac{1}{|\alpha_{j,\ell}^{(m)}|} \\
&\geq \sum_{j=1}^{s} m_j \log  \frac{1}{|a_j| + \rho_j}\\
&> R - \delta/2 \\
&\geq R - \delta,
\end{aligned}
\end{equation}
for all sufficiently large $m$. Choose one index $m$ for which both \eqref{eq:power-bound-pw-final} and \eqref{eq:rate-bound-pw-final} hold, and set $\widehat B=B_m$. This proves \eqref{lemma:pw-two-bounds}.
\end{proof}

\section{Proof of Lemma~\ref{lemma-Schur-completion}}\label{appendix:lemma-Schur-completion}
In Subsection~\ref{app-sub:schur_reduction_reconstruction_and_degree_formulas}, we recall the Schur reduction and reconstruction operators at an arbitrary point in $\mathbb D$ together with their degree properties for finite Blaschke products. Subsection~\ref{app-sub:proof_of_lemma_ref_lemma_schur_completion} then uses these transformations to prove Lemma~\ref{lemma-Schur-completion} in five steps.

\subsection{Schur Reduction, Reconstruction, and Degree Formulas}\label{app-sub:schur_reduction_reconstruction_and_degree_formulas}
Fix $a,w\in\mathbb D$, and let $b_a$ be as in Section~\ref{subsection:preliminaries}. For a Schur function $f$ satisfying $f(a)=w$ and an arbitrary Schur function $g$, define
$$
(R_{a,w}f)(z) \triangleq \frac{f(z)-w}{b_a(z)(1-\overline{w}f(z))},\qquad (T_{a,w}g)(z) \triangleq\frac{w + b_a(z) g(z)}{1 + \overline{w} b_a(z) g(z)},
$$
where the value of $R_{a,w}f$ at $a$ is defined by continuity. It follows from the Schwarz--Pick theorem (see, e.g., \cite[Theorem 1.4.1]{garcia2018finite}) that $R_{a,w}f$ is a Schur function. In addition, since $b_a(z)g(z)$ is a Schur function and the map $z\mapsto(w+z)/(1+\overline wz)$ is a disk automorphism (i.e., a bijective analytic map from $\mathbb D$ onto $\mathbb D$), $T_{a,w}g$ is also a Schur function and satisfies $(T_{a,w}g)(a)=w$. By direct substitution, we have
\begin{equation*}\label{pf-lemma-Schur-completion-TR}
T_{a,w}(R_{a,w}f)=f,\qquad R_{a,w}(T_{a,w}g)=g.
\end{equation*}
In particular, when $a=0$, these formulas reduce to the usual Schur recursion and its inverse \cite[Subsection 1.3.6]{simon2005}. 

We also use the following degree properties when $f$ and $g$ are finite Blaschke products (see \cite[Definition~2.4 and Proposition~2.5]{2021Abate}): for fixed $a,w\in\mathbb D$, if $f$ is a finite Blaschke product satisfying $f(a)=w$ (so $\deg f\ge 1)$ and $g$ is an arbitrary finite Blaschke product, then $R_{a,w}f$ and $T_{a,w}g$ are finite Blaschke products satisfying
\begin{equation}\label{eq:Schur-completion-degree}
\deg(R_{a,w}f)=\deg f-1, \qquad \deg(T_{a,w}g)=\deg g+1.
\end{equation}
\noindent\emph{Derivation of~\eqref{eq:Schur-completion-degree}:} For completeness, we give a direct proof of \eqref{eq:Schur-completion-degree}. Let $\Theta$ be any finite Blaschke product of degree $d$ and let $c\in\mathbb D$. Define
\begin{equation}\label{preliminary-def-Psi}
\Psi_c\triangleq\frac{\Theta-c}{1-\overline c\Theta}.
\end{equation}
Both $\Theta$ and $\Theta-c$ are analytic in a neighborhood of $\overline{\mathbb D}$, and
$$
\left|\Theta(z)-(\Theta(z)-c)\right|=|c|<1=|\Theta(z)|, \qquad z\in \mathbb T.
$$
By Rouch\'e's theorem \cite[Theorem 10.43(b)]{rudin1987}, $\Theta-c$ has exactly $d$ zeros $\zeta_1,\ldots,\zeta_d$ in $\mathbb D$, counted with multiplicity. Since $|1-\overline c\Theta|\ge1-|c|>0$ on $\mathbb D$, the points $\zeta_1,\ldots,\zeta_d$ are also precisely the zeros of $\Psi_c$. Moreover, $|\Psi_c|=1$ on $\mathbb T$. Let
$$
h\triangleq\frac{\Psi_c}{\prod_{r=1}^d b_{\zeta_r}}.
$$ 
Then $h$ can be extended analytically and without zeros to $\mathbb D$. Thus, both $h$ and $1/h$ are continuous on $\overline{\mathbb D}$ with $|h|=1$ on $\mathbb T$. By the maximum modulus principle \cite[Theorem 10.24 and Section 12.1]{rudin1987}, we have
$$
|h(z)|\le1,\qquad |h(z)|^{-1}\le1, \qquad z\in\mathbb D.
$$
Hence, $h(z)\equiv\lambda$ for some $\lambda\in\mathbb T$, and $\Psi_c=\lambda \prod_{r=1}^d b_{\zeta_r}$ is a finite Blaschke product of degree $d$. 

Taking $\Theta=f$ and $c=w$ in \eqref{preliminary-def-Psi} implies $\Psi_w(a)=0$, and hence
$$
R_{a,w}f=\frac{\Psi_w}{b_a},\qquad \deg(R_{a,w}f)=\deg f-1.
$$
Similarly, taking $\Theta=b_ag$ and $c=-w$ in \eqref{preliminary-def-Psi} implies
$$
\Psi_{-w} =\frac{b_ag+w}{1+\overline w b_ag} =T_{a,w}g.
$$
Since $\deg(b_ag)=\deg g+1$, we obtain
$$
\deg(T_{a,w}g)=\deg \Psi_{-w}=\deg (b_ag)=\deg g+1,
$$
which establishes \eqref{eq:Schur-completion-degree}.

\subsection{Proof of Lemma~\ref{lemma-Schur-completion}}\label{app-sub:proof_of_lemma_ref_lemma_schur_completion}
\begin{proof}
\textbf{Step 1:} \emph{Iterative Schur Reduction.} Let $N\triangleq\sum_{i=1}^nm_i$. List each $\alpha_i$ exactly $m_i$ times as $\gamma_1,\ldots,\gamma_N$, and set $\varphi_0\triangleq\varphi$. At step $j=1,\ldots,N$, let $w_j\triangleq\varphi_{j-1}(\gamma_j)$ and proceed as follows:
\begin{equation}\label{eq:lemma-Schur-completion-recursion}
\begin{cases}
\varphi_j=R_{\gamma_j,w_j}\varphi_{j-1}, & |w_j|<1,\\[3pt] 
\varphi_{j-1}\equiv w_j~(\text{and terminate the iteration}), & |w_j|=1,
\end{cases}
\end{equation}
where we have applied the maximum modulus principle in the case $|w_j|=1$. In this degenerate case when $|w_j|=1$, reversing the preceding $j-1$ inverse transformations shows that $\varphi$ is a finite Blaschke product of degree $j-1\le N$; taking $\Theta=\varphi$ proves the lemma immediately. Assume henceforth that $|w_j|<1$ for all $j=1,\ldots,N$, and abbreviate
$$
R_j\triangleq R_{\gamma_j,w_j}, \qquad T_j\triangleq T_{\gamma_j,w_j}.
$$ 
Then, $\varphi_{j-1}=T_j(\varphi_j)$ for every $j=1,\ldots,N$.

\medskip
\textbf{Step 2:} \emph{Derivative Propagation and Difference Factorization.} Fix $j\in\{1,\ldots,N\}$. We first show that for Schur functions $f,h$ satisfying $f(\gamma_j)=h(\gamma_j)=w_j$, any $\xi\in\mathbb D$ and any integer $m\ge1$, the condition
\begin{equation}\label{eq:lemma-Schur-completion-matching}
f^{(\ell)}(\xi)=h^{(\ell)}(\xi),\qquad 0\le\ell<m,
\end{equation}
implies
\begin{equation}\label{eq:lemma-Schur-completion-reduced-matching}
\begin{cases}
(R_jf)^{(\ell)}(\xi)=(R_jh)^{(\ell)}(\xi), &0\le\ell<m-1,\quad \xi=\gamma_j,\\[3pt]
(R_jf)^{(\ell)}(\xi)=(R_jh)^{(\ell)}(\xi), &0\le\ell<m,\quad \xi\ne\gamma_j.
\end{cases}
\end{equation}
To see this, note that for $z\in\mathbb D\setminus\{\gamma_j\}$, 
\begin{equation}\label{eq:lemma-Schur-completion-forward}
(R_jf)(z)-(R_jh)(z)=\frac{q_j(z)}{z-\gamma_j}\bigl(f(z)-h(z)\bigr),
\end{equation}
where
$$
q_j(z) \triangleq\frac{(1-|w_j|^2)(1-\overline{\gamma_j}z)}{(1-\overline{w_j}f(z))(1-\overline{w_j}h(z))}.
$$
Since
$$
\min \left\{|1-\overline{w_j}f(z)|,~|1-\overline{w_j}h(z)| \right\}\ge1-|w_j|>0, \qquad z\in\mathbb D,
$$
the function $q_j(z)$ is analytic in $\mathbb D$. Moreover, by \eqref{eq:lemma-Schur-completion-matching}, the quotient
$$
p(z)\triangleq\frac{f(z)-h(z)}{(z-\xi)^m}, \qquad z\in\mathbb D\setminus\{\xi\},
$$
extends analytically to $\mathbb D$. Consequently,
$$
f(z)-h(z)=(z-\xi)^m p(z), \qquad z\in\mathbb D,
$$
which, substituted into \eqref{eq:lemma-Schur-completion-forward}, immediately implies that
$$
(R_jf)(z)-(R_jh)(z)
=
\begin{cases}
(z-\xi)^{m-1}q_j(z)p(z), &\xi=\gamma_j,\\[4pt]
(z-\xi)^m\dfrac{q_j(z)p(z)}{z-\gamma_j}, &\xi\ne\gamma_j.
\end{cases}
$$
Both $q_j(z)p(z)$ and, when $\xi\ne\gamma_j$, $q_j(z)p(z)/(z-\gamma_j)$ are analytic near $\xi$, so \eqref{eq:lemma-Schur-completion-reduced-matching} follows. 

On the other hand, for arbitrary Schur functions $u,v$, it holds that
\begin{equation}\label{eq:lemma-Schur-completion-inverse}
T_j(u)-T_j(v) = b_{\gamma_j}(u-v)L_j(u,v),
\end{equation}
where 
$$
L_j(u,v)\triangleq\frac{1-|w_j|^2} {(1+\overline{w_j}b_{\gamma_j}u) (1+\overline{w_j}b_{\gamma_j}v)}.
$$
For every $z\in\mathbb D$, we have
$$
\min \left\{ |1+\overline{w_j}b_{\gamma_j}(z)u(z)|,~ |1+\overline{w_j}b_{\gamma_j}(z)v(z)| \right\}\ge1-|w_j|>0,
$$
and consequently
\begin{equation}\label{eq:lemma-Schur-completion-L-j}
L_j(u,v)\in \mathcal{H}^\infty,\qquad \|L_j(u,v) \|_\infty\le \frac{1-|w_j|^2}{(1-|w_j|)^2}= \frac{1+|w_j|}{1-|w_j|}.
\end{equation}

\medskip
\textbf{Step 3:} \emph{Parameterized Family of Blaschke Interpolants $\Theta_\tau$.} For each $\tau\in\mathbb T$, define
$$
U_N^\tau\equiv\tau,\qquad U_{j-1}^\tau\triangleq T_j(U_j^\tau) \quad (j=N,N-1,\ldots,1),\quad \Theta_\tau=U_0^\tau.
$$
It then follows from \eqref{eq:Schur-completion-degree} that $\deg U_j^\tau=N-j$ for all $0\le j\le N$; in particular, $\Theta_\tau$ is a finite Blaschke product of degree $N$. Moreover, it follows from $\varphi_{j-1}=T_j(\varphi_j)$ and \eqref{eq:lemma-Schur-completion-inverse} that
\begin{equation}\label{eq:lem3-3-recursive-formula}
U_{j-1}^\tau-\varphi_{j-1} = b_{\gamma_j}K_j^\tau(U_j^\tau-\varphi_j),
\end{equation}
where $K_j^\tau\triangleq L_j(U^\tau_j,\varphi_j)$. Iterating \eqref{eq:lem3-3-recursive-formula} from $j=N$ down to $j=1$, we have
\begin{equation}\label{eq:lemma-Schur-completion-error-product}
\Theta_\tau-\varphi= \left(\prod_{j=1}^N b_{\gamma_j}\right) \left(\prod_{j=1}^N K_j^\tau\right)(\tau-\varphi_N).
\end{equation}
By \eqref{eq:lemma-Schur-completion-L-j}, $K_j^\tau\in\mathcal{H}^\infty$ and $\|\tau-\varphi_N\|_\infty\le2$. Since $\beta=\eta\prod_{j=1}^N b_{\gamma_j}$ for some $\eta\in\mathbb T$, we obtain 
\begin{equation}\label{lem-3-3-pf-H-infty}
\frac{\Theta_\tau-\varphi}{\beta} =  \eta^{-1}\left(\prod_{j=1}^N K_j^\tau\right)(\tau-\varphi_N)\in \mathcal H^\infty.
\end{equation}
Thus, every $\Theta_\tau$ satisfies all the interpolation conditions in the lemma, namely, 
$$
\Theta_\tau^{(\ell)}(\alpha_i)=\varphi^{(\ell)}(\alpha_i),\qquad 0\le\ell<m_i,\quad 1\le i\le n.
$$

\medskip
\textbf{Step 4:} \emph{Existence and Uniqueness under Boundary Normalization.} To obtain real Taylor coefficients, we impose the additional normalization $\Psi(1)=1$. We claim that there exists a unique finite Blaschke product $\Psi$ satisfying
\begin{equation}\label{lem:pr-any-Psi}
\deg\Psi\le N,\qquad \Psi(1)=1,\qquad \Psi^{(\ell)}(\alpha_i)=\varphi^{(\ell)}(\alpha_i), \quad 0\le\ell<m_i,\quad 1\le i\le n.
\end{equation}

To prove existence, for $j=1,\ldots,N$, define
$$
A_j(z) \triangleq \frac{w_j+b_{\gamma_j}(1)z} {1+\overline{w_j}b_{\gamma_j}(1)z}, \qquad z\in\overline{\mathbb D}.
$$
Since $|b_{\gamma_j}(1)|=1$ and $|w_j|<1$, each $A_j$ is a disk automorphism and maps $\mathbb T$ bijectively onto itself \cite[Theorem 12.4]{rudin1987}. Evaluating $U_{j-1}^\tau=T_j(U_j^\tau)$ at $z=1$ yields
$$
U_{j-1}^\tau(1)=A_j\bigl(U_j^\tau(1)\bigr), \qquad j=1,\ldots,N,
$$
which, together with $U_N^\tau(1)=\tau$, implies that
\begin{equation}\label{eq:lemma-Schur-completion-boundary}
\Theta_\tau(1)=A(\tau), \qquad A\triangleq A_1\circ\cdots\circ A_N.
\end{equation}
Since $A$ is a bijection of $\mathbb T$ onto itself, setting
$$
\tau\triangleq A^{-1}(1),\qquad \Theta\triangleq\Theta_\tau,
$$
establishes $\Theta(1)=1$. Together with Step 3, this shows that $\Theta$ satisfies \eqref{lem:pr-any-Psi}.

To prove uniqueness, let $\Psi$ be any finite Blaschke product satisfying \eqref{lem:pr-any-Psi}, and let $d\triangleq\deg\Psi$ and $\Psi_0\triangleq\Psi$. For $1\le i\le n$ and $0\le j\le N$, define
$$
n_i(j)\triangleq\#\{r:1\le r\le j,\ \gamma_r=\alpha_i\}.
$$
We show inductively that the reductions $\Psi_j=R_j\Psi_{j-1}$ are well-defined and satisfy
\begin{equation}\label{eq:lemma-Schur-completion-remaining}
\Psi_j^{(\ell)}(\alpha_i)=\varphi_j^{(\ell)}(\alpha_i), \qquad 0\le\ell<m_i-n_i(j),\quad 1\le i\le n.
\end{equation}
For $j=0$, this follows from \eqref{lem:pr-any-Psi}. Suppose that \eqref{eq:lemma-Schur-completion-remaining} holds at $j-1$, and let $\gamma_j=\alpha_i$.
Since $n_i(j-1)<m_i$, we have
$$
\Psi_{j-1}(\gamma_j) =\varphi_{j-1}(\gamma_j)=w_j\in\mathbb D.
$$
Hence, $\Psi_j=R_j\Psi_{j-1}$ is well-defined, and \eqref{eq:lemma-Schur-completion-reduced-matching} implies \eqref{eq:lemma-Schur-completion-remaining} at step $j$. By \eqref{eq:Schur-completion-degree}, each $\Psi_j$ is a finite Blaschke product with $\deg\Psi_j=d-j$. Therefore,
$$
0\le\deg\Psi_N=d-N\le0,
$$
which implies that $d=N$ and $\Psi_N\equiv\sigma$ for some $\sigma\in\mathbb T$. Reversing the reductions implies $\Psi=T_1\circ\cdots\circ T_N(\sigma)=\Theta_\sigma$. By \eqref{eq:lemma-Schur-completion-boundary},
\begin{equation}\label{lem-3-3-pr-A}
A(\sigma)=\Psi(1)=1=A(\tau).
\end{equation}
Since $A$ is injective on $\mathbb T$, we obtain $\sigma=\tau$ and $\Psi=\Theta$, proving uniqueness in \eqref{lem:pr-any-Psi}.

\medskip
\textbf{Step 5:} \emph{Verification of Real Taylor Coefficients.} Define $\Theta^\#(z) \triangleq \overline{\Theta(\overline{z})}$. Then $\Theta^\#$ is a finite Blaschke product of degree $N$. Since the interpolation points occur in conjugate pairs with equal multiplicities, for each $1\le i\le n$ and $0\le\ell<m_i$,
$$
(\Theta^\#)^{(\ell)}(\alpha_i) =\overline{\Theta^{(\ell)}(\overline{\alpha_i})} =\overline{\varphi^{(\ell)}(\overline{\alpha_i})} =\varphi^{(\ell)}(\alpha_i),
$$
where the last equality follows from the real Taylor coefficients of $\varphi$. Moreover, $\Theta^\#(1)=\overline{\Theta(1)}=1$. Thus, $\Theta^\#$ satisfies \eqref{lem:pr-any-Psi}. By the uniqueness established in Step 4, $\Theta^\#=\Theta$, so $\Theta$ has real Taylor coefficients. Together with \eqref{lem-3-3-pf-H-infty}, this completes the proof.
\end{proof}

\bibliographystyle{ieeetr}
\bibliography{Kim_gap}

\end{document}